\documentclass[a4paper]{article}
\usepackage{amsmath,amssymb,amsthm,mathtools,cancel,mathdots,adjustbox,mathtools,color,tikz,todonotes,graphicx,framed}
\usepackage[utf8]{inputenc}
\usepackage[T1]{fontenc}
\usepackage{eufrak}
\usepackage{geometry}
\usepackage{amsmath}
\usepackage{amssymb}
\usepackage{amsthm}
\usepackage{makeidx}
\usepackage{mathtools}
\usepackage{mathabx, mathrsfs, dsfont}
\usepackage{cases}
\usepackage{braket}
\usepackage[toc,page]{appendix}
\usepackage{pdfsync}
\usepackage{hyperref}
\usepackage{graphicx}
\usepackage{psfrag}
\usepackage{epstopdf}
\usepackage{graphicx}
\usepackage{fancyhdr}
\usepackage{math}
\usepackage{cite}
\usepackage{color}

\renewcommand{\meanval}[1]{\bE\left[#1\right]}

\renewcommand{\mod}{{\,{\rm mod}\,}}

\newtheorem{lemma}{Lemma}[section]
\newtheorem{theorem}[lemma]{Theorem}
\newtheorem{proposition}[lemma]{Proposition}
\newtheorem{corollary}[lemma]{Corollary}
\newtheorem{remark}[lemma]{Remark}

\newtheorem{definition}[lemma]{Definition}

\numberwithin{equation}{section}

\title{Adiabatic invariants for the FPUT and Toda chain in the thermodynamic limit}
\author{
T. Grava
\footnote{
International School for Advanced Studies (SISSA), Via Bonomea 265,   34136, Trieste, Italy and School of Mathematics,  University of Bristol,  Fry Building,  BS8 1UG,   UK\newline
 \textit{Email: } \texttt{grava@sissa.it} 
}  ,
 A. Maspero
 \footnote{
International School for Advanced Studies (SISSA), Via Bonomea 265,  34136  Trieste, Italy \newline
 \textit{Email: } \texttt{alberto.maspero@sissa.it} 
 }  ,
 G. Mazzuca
 \footnote{
International School for Advanced Studies (SISSA), Via Bonomea 265,  34136 Trieste, Italy \newline
 \textit{Email: } \texttt{guido.mazzuca@sissa.it} 
 }  ,
 A. Ponno
 \footnote{
Universit\'a di Padova, Dipartimento di Matematica ''Tullio Levi-Civita``, Via Trieste 63, 35121, Padova,  Italy
 \newline
 \textit{Email: } \texttt{ponno@math.unipd.it} 
 }
}
\date{\today}

\begin{document}

\maketitle

\begin{abstract}
We consider the  Fermi-Pasta-Ulam-Tsingou   (FPUT)  chain composed by $N \gg 1$ particles and periodic boundary conditions, and endow the phase space with the Gibbs measure at small temperature $\beta^{-1}$. 
Given a fixed ${1\leq m \ll N}$, we prove that the first $m$  integrals  of motion of the periodic Toda chain are adiabatic  invariants of FPUT (namely they are approximately constant along the Hamiltonian flow of the FPUT) for times of order $\beta$,  for initial data in a set of large measure.  \\
We also prove that special linear combinations of the harmonic energies are adiabatic invariants of the FPUT on the same time scale, whereas they become adiabatic invariants for all times for the Toda dynamics.
\end{abstract}

{\bf Keywords: } Fermi-Pasta-Ulam-Tsingou  chain, Toda chain, adiabatic invariants, Gibbs measure.
\section{Introduction and Main Results}
The FPUT chain with $N$ particles is the system with  Hamiltonian  
\begin{equation}
H_{F}(\bp,\bq)= \sum_{j=0}^{N-1}\frac{p_j^2}{2} + \sum_{j=0}^{N-1} V_F(q_{j+1}-q_{j})  \ , \quad V_F(x) = \frac{x^2}{2} -   \frac{x^3}{6}+ \tb  \frac{x^4}{24} \ ,
\label{fpu}
\end{equation}
which we consider with periodic boundary conditions $q_N=q_0\,$, $p_N=p_0$ and  $\tb > 0$.  We observe  that any generic nearest neighborhood quartic potential can  be set in the form of $V_F(x)$ through a canonical change of coordinates.
{\color{black}Over the last 60 years the FPUT system   has been the object of intense numerical and analytical research. Nowadays it is well understood that the system displays, on a relatively short time scale, an integrable-like behavior, first uncovered by Fermi, Pasta, Ulam and Tsingou \cite{Fermi},\cite{Fermi0} and later interpreted in terms of closeness to a nonlinear integrable
system by some authors, e.g. the Korteweg-de Vries (KdV) equation by Zabusky and Kruskal \cite{ZK}, the Boussinesq equation by Zakharov \cite{Zakharov}, and the Toda chain by Manakov first \cite{Manakov}, and then by Ferguson, Flaschka and McLauglin \cite{ferguson}. On larger time scales the system displays instead an ergodic behavior and approaches its micro-canonical equilibrium state (i.e. measure), unless the energy is so low to enter a
KAM-like regime \cite{Chirikov,HK,RINK}.

 In the present work  we  show that a family of first integrals of the Toda system are  adiabatic invariants (namely almost constant quantities) for the FPUT system. 
We bound their variation for times of order $\beta$,   where $\beta$ is the inverse of the temperature of the chain. 
Such estimates hold for a large set of initial data with respect to the Gibbs measure  of the chain  and they are  uniform in the number of particles, thus they persist in the thermodynamic limit.

In the last few years, there has been   a lot of activity  in the problem of constructing  adiabatic invariants  of nonlinear  chain  systems  in the thermodynamic limit, 
see \cite{carati,CaratiMaiocchi,BCM14,giorgilli12,giorgilli15,maiocchi}.
In particular adiabatic invariants in measure for the FPUT chain have been recently introduced by Maiocchi, Bambusi, Carati \cite{BCM14} by considering the FPUT chain a perturbation of the linear harmonic chain.
 Our approach  is based on the remark \cite{Manakov,ferguson} that the FPUT chain 
\eqref{fpu} can be regarded as a perturbation of  the  (nonlinear)  Toda chain  \cite{toda}
\begin{equation}
H_{T}(\bp,\bq):=\frac{1}{2} \sum_{j=0}^{N-1}{p_j^2} + \sum_{j=0}^{N-1}V_T(q_{j+1}-q_{j} )\ , \quad  V_T(x) =  e^{- x} +  x - 1\, ,
\label{toda}
\end{equation}
which we consider again with periodic boundary conditions $q_N=q_0\,$, $p_N=p_0$.
{The equations of motion of \eqref{fpu} and \eqref{toda} take the form
\begin{equation}
\label{Todaeq}
\dot{q}_j=\dfrac{\partial H}{\partial p_j}=p_j,\quad \dot{p}_j=-\dfrac{\partial H}{\partial q_j}=V'(q_{j+1}-q_{j}) -V'(q_{j}-q_{j-1}), \;\;j=0,\dots,N-1,
\end{equation}}
where $H$ stands for  $H_F$ or $H_T$ and  $V$ for  $V_F$ and $V_T$ respectively.\\
According to the values of $\tb$ in \eqref{fpu}, the Toda chain is either an approximation of the FPUT chain of   third order (for $\tb \neq 1$), or  fourth order (for $\tb = 1$).
We remark that the Toda chain  is the only  nonlinear integrable FPUT-like chain \cite{Sako,Dubrovin08}. 

%
The Toda chain admits {several families of } $N$ integrals of motion in involution (e.g. \cite{VanMoerbeke1976,FM,HK2}). 
Among  the various families of  integrals  of motion, the ones  constructed by Henon \cite{Henon} and Flaschka \cite{Flaschka1974} are explicit and easy to compute, being the trace of the powers of the Lax matrix associated to the Toda chain. In the following we refer to them simply as {\em Toda integrals} and denote them by  $J^{(k)}$, $1 \leq k \leq N$ (see \eqref{integrals}).

As the $J^{(k)}$'s are conserved along the Toda flow, and the FPUT chain  is a perturbation of the Toda one, the Toda integrals  are good candidates to be adiabatic invariants when computed along the FPUT flow.
This intuition is supported by several numerical simulations, the first by  Ferguson-Flaschka-McLaughlin \cite{ferguson} and more recently by other authors \cite{PCSF, benettin-ponno11,benettin-ponno13,GK,ChEh19}.
Such simulations  show that the variation of the Toda integrals  along the FPUT flow is very small on long times  for initial data of small specific energy. In particular, the numerical results in \cite{benettin-ponno11,benettin-ponno13,GK} suggest that  such phenomenon  should persists in the thermodynamic limit and for ``generic'' initial conditions.

Our first result is a quantitative,  analytical proof of this phenomenon. 
More precisely, we fix an arbitrary  $m \in \N$  and 
provided $N$ and $\beta$    sufficiently large,  we bound the variations of the   first $m$ Toda integrals computed along the flow of FPUT,    for  times of order 
\begin{equation}
\label{bound}
\frac{\beta}{\left((\tb-1)^2 + C_1 \beta^{-1} \right)^{\frac12}} ,
\end{equation}
where  $C_1$ is a positive constant, independent of $\beta, N$.
Such a  bound holds for initial data
in a large set with respect to the  Gibbs measure. Note that the bound \eqref{bound}  improves to $\beta^{\frac{3}{2}}$ when  $\tb = 1$, namely when  the Toda chain becomes a fourth order approximation of the FPUT chain. 
Such analytical time-scales are compatible with (namely smaller than) the numerical ones determined in
\cite{benettin-ponno11,benettin-ponno13,BPP18}.}

An interesting question is whether the Toda integrals $J^{(k)}$'s control  the normal modes of FPUT, namely the action of the linearized chain.
It turns out that this is indeed the case: 
we prove that the quadratic parts  $J^{(2k)}_2$  (namely the Taylor polynomials of order 2) of the  integral of motions  $J^{(2k)}$, are linear combinations of the normal modes. Namely  one has 

\begin{equation}
\label{struc}
J^{(2k)} = \sum_{j=0}^{N-1} \wh c_j^{(k)}\,  E_j  + O((\wh \bp, \wh \bq)^3) , 
\end{equation}
where $E_j$ is the $j^{th}$ normal mode  (see \eqref{he} for its formula), $(\wh \bp, \wh \bq)$ are the discrete Hartley transform  of $(\bp,\bq)$ (see  definition below in \eqref{dht}) and  $\wh\bc^{(k)}$ are real coefficients.

So we consider  linear combinations of the normal modes  of the form
\begin{equation}
\label{lin.comb}
\sum_{j=0}^{N-1} \wh g_j  E_j
\end{equation}
where $(\wh g_j)_j$ is the discrete Hartley transform of  a vector $\bg \in \R^N$ which has only   $2\floor{\frac{m}{2}}+2$ nonzero entries with $m$ independent from $N$, here $\floor{\frac{m}{2}}$ is the integer part of $\frac{m}{2}$. 
Our second result shows that  linear combinations of the form \eqref{lin.comb}, when computed along the FPUT flow,  are  
adiabatic invariants  for the same time scale as in \eqref{bound}.\\
Further we also show that linear combinations of the harmonic modes as  in  \eqref{lin.comb}, are approximate invariant for the Toda dynamics (with large probability).

Examples of linear combinations \eqref{lin.comb} that we control are
\begin{equation}
\label{sincos}
\sum_{j=1}^N \sin^{2\ell}\left(\frac{j\pi}{N}\right) \, E_j , \qquad 
\sum_{j=1}^N \cos^{2\ell}\left(\frac{j\pi}{N}\right) \, E_j , \qquad \forall \ell=0, \ldots, \Big\lfloor\frac{m}{2}\Big\rfloor .  
\end{equation}
These linear combinations weight in different ways low and high energy modes.

{
Finally we note that, in the  study of the FPUT problem, one usually measures the time the system takes to approach the equilibrium  when   initial conditions very far from equilibrium are considered.
On the other hand, 
our result indicates that, despite initial states
are sampled from a thermal distribution, nonetheless complete thermalization
is expected to be attained, in principle, over a time scale that increases
with decreasing temperature.}\\
Our results are mainly based on two ingredients. The first one is a detailed study of the algebraic properties of the Toda integrals. 
The second ingredient comes from adapting  to our case, methods of statistical mechanics developed  by Carati \cite{carati} and Carati-Maiocchi \cite{CaratiMaiocchi}, and also in \cite{BCM14,giorgilli12,giorgilli15,maiocchi}. 

\bigskip

\section{Statement of results}
\subsection{Toda integrals as adiabatic invariants for FPUT}

We come to a precise statements of the main results of the present paper.
 We consider the FPUT chain \eqref{fpu} and the Toda chain \eqref{toda} in the subspace 
\begin{equation}
\label{m_qp}
\cM:= \left\lbrace (\bp,\bq) \in \R^{N}\times \R^N \colon \ \ \ 
\sum_{j=0}^{N-1} q_j ={\mathcal  L}  \ \, ,\sum_{j=0}^{N-1} p_j = 0  \right\rbrace \, ,
\end{equation}
which is invariant for the dynamics.  Here  ${\mathcal L}$ is a positive constant. 

Since both $H_F$ and $H_T$ depend just on the relative distance  between $q_{j+1}$ and $q_j$, it is natural to introduce on $\cM$ the variables $r_j$'s as
\begin{equation}
\label{eq:r_change}
	r_j := q_{j+1} - q_j , \qquad 0 \leq j \leq N -1 \, ,
\end{equation}
which are naturally constrained to
\begin{equation}
\label{sumr}
\sum_{j=0}^{N-1} r_j = 0 \, ,
\end{equation} 
due to the periodic boundary condition $q_N = q_0$. We observe  that the change of coordinates \eqref{eq:r_change} together with the condition \eqref{sumr} is well defined on the phase space $\cM$, but not on the whole  phase space  $\R^N\times \R^N$.
In these variables the phase space $\cM$ reads
\begin{equation}
\label{media}
\cM:= \left\lbrace (\bp,\br) \in \R^{N}\times \R^N \colon \ \ \ 
\sum_{j=0}^{N-1} r_j = \sum_{j=0}^{N-1} p_j = 0  \right\rbrace \, .
\end{equation}
We endow $\cM$  by the Gibbs measure of $H_F$  at temperature $\beta^{-1}$, namely we put
\begin{equation}\label{eq:misura_vera}
\di \mu_F :=    \frac{1}{Z_F(\beta)} \  e^{-\beta {H_F(\bp,\br)}} \,  \ \delta\left(\sum_{j=0}^{N-1} r_j =0\right) \ \delta\left(\sum_{j=0}^{N-1} p_j =0\right)  \ \di \bp \,  \di \br, 
\end{equation}
where
  as usual $Z_F(\beta)$ is the partition function which normalize the measure, namely
  \begin{equation}
Z_F(\beta) := \int_{\R^N \times \R^N}    e^{-\beta{H_F(\bp,\br)}} \ \delta\left(\sum_{j=0}^{N-1} r_j =0\right) \ \delta\left(\sum_{j=0}^{N-1} p_j =0\right)  \ \di \bp \,  \di \br . 
\end{equation}
{We remark that we can consider the measure $\di \mu_F$ as the weak limit, as $\epsilon \to 0$, of the measure 
\begin{equation*}
	\di \mu_\epsilon = 
	\frac{
	e^{-\beta H_F(\bp, \br)} \ 
	e^{-\left(\sum_{j=0}^{N-1}r_j/\epsilon \right)^2  
	    - \left(\sum_{j=0}^{N-1}p_j/\epsilon \right)^2   }}{\left(\int_{\R^{2N}}e^{-\beta H_F(\bp, \br)} \ 
	e^{-\left(\sum_{j=0}^{N-1}r_j/\epsilon \right)^2  
	    - \left(\sum_{j=0}^{N-1}p_j/\epsilon \right)^2   } \  \di \bp \, \di \br \right) } \  \di \bp \, \di \br \, . 
\end{equation*} 
}
Given a function $f\colon \cM \to \C$, we will use the probability \eqref{eq:misura_vera}  to compute its average $\la f \ra$, its $L^2$ norm $\norm{f}$, its variance $\sigma_f^2$ defined as
\begin{align}
\label{average}
    & \la f \ra := \meanval{f} \equiv \int_{\R^{2N}} f(\bp,\br)  \, \, \di \mu_F , \\
    & \norm{ f}^2 := \meanval{|f|^2} \equiv \int_{\R^{2N}} |f(\bp,\br)|^2 \, \di \mu_F , \\
    & \sigma_f^2 :=   \norm{f - \la f \ra}^2 . 
\end{align}

In order to state our first theorem we must introduce the Toda integrals of motion. It is well known that the Toda chain is an integrable system \cite{toda,Henon}.
The standard way to prove its  integrability is to put it in a Lax-pair form. 
The  Lax form was introduced by  Flaschka in  \cite{Flaschka1974} and Manakov \cite{Manakov} and it is obtained through the change of coordinates 
\begin{equation}
b_j := -p_j \, , \qquad a_j:=  e^{\frac{1}{2}(q_j-q_{j+1})} \equiv e^{- \frac{1}{2} r_j} , \qquad 0 \leq j \leq N-1 \, .
\label{bavariable}
\end{equation}
By the geometric constraint \eqref{sumr} and the momentum conservation $\sum_{j=0}^{N-1} p_j = 0$ (see \eqref{m_qp}), such variables are  constrained by the conditions 
\begin{equation*}
    \sum_{j=0}^{N-1}{b_j}=0, \,  \qquad \prod_{j=0}^{N-1}{a_j}=1 \ .
    \end{equation*}
The Lax operator for the Toda chain is the periodic Jacobi matrix
 \cite{VanMoerbeke1976}
\begin{equation} \label{jacobi}
L(b,a) := \left( \begin{array}{ccccc}
b_{0} & a_{0} & 0 & \ldots &  a_{N-1} \\
a_{0} & b_{1} & a_{1} & \ddots & \vdots \\
0 & a_{1} & b_{2} & \ddots & 0 \\
\vdots & \ddots & \ddots & \ddots & a_{N-2} \\
 a_{N-1} & \ldots & 0 & a_{N-2} & b_{N-1} \\
\end{array} \right) .
\end{equation}
We introduce  the matrix  $A=L_+-L_-$  where  for a square matrix
 $X$  we call  $X_+$  the upper triangular part of $X$ 
$$
\left(X_+\right)_{ij} =\left\{ \begin{array}{cc} X_{ij}, & i\leq j \\
0, & \mbox{otherwise}\end{array}\right.
$$
and in a similar way by $X_-$ the lower triangular part of $X$
$$
\left(X_-\right)_{ij} =\left\{ \begin{array}{cc} X_{ij} , & i\geq j \\
0, & \mbox{otherwise.}\end{array}\right.
$$
A straightforward calculation shows that  the Toda equations of motions \eqref{Todaeq} are equivalent to 
$$ \dfrac{d L}{dt}= [A,L].$$
It then follows that the eigenvalues of $L$ are  integrals of motion in involutions.

In particular,  the trace of powers of $L$, 
\begin{equation}
\label{integrals}
J^{(m)} := \frac{1}{m} \tr{L^m} , \qquad \forall 1 \leq m \leq N 
\end{equation}
are  $N$ independent, commuting,  integrals of motions in involution.
 Such integrals were first introduced by  Henon \cite{Henon} (with a different method), and we refer to them as {\em Toda integrals}.  We give the first few of them explicitly, written in the variables $(\bp, \br)$:
 \begin{equation}
 \label{t.int}
\begin{aligned}
&J^{(1)}(\bp) := -\sum_{i=0}^{N-1} p_i,  \qquad \qquad 
 J^{(2)}(\bp,\br):= \sum_{i=0}^{N-1}\left[ \frac{p_i^2}{2} + e^{-r_i}\right], \\
 & J^{(3)}(\bp,\br):= -\sum_{i=0}^{N-1} \left[\frac{1}{3} p_i^3 + (p_i + p_{i+1}) e^{-r_i}\right],  \\
&J^{(4)}(\bp,\br):= \sum_{i=0}^{N-1} \left[\frac{1}{4} p_i^4 + (p_i^2 + p_i p_{i+1} + p_{i+1}^2) e^{-r_i} + \frac{1}{2} e^{-2r_i} + e^{-r_i - r_{i+1}}\right] . 
\end{aligned}
\end{equation}
Note that $J^{(2)}$ coincides with the Toda Hamiltonian $H_T$. \\

Our first result shows that the Toda integral $J^{(m)}$, computed along the Hamiltonian flow $\phi^t_{H_F}$ of the FPUT chain, is an  adiabatic invariant  for long times and  for a set of initial data in a set of large Gibbs measure.
Here the precise statement:
\begin{theorem}\label{thm:goal}
Fix $m \in \N$.  There exist constants  $N_0, \beta_0, C_0, C_1>0$ (depending on $m$),  such that for any $N > N_0$, $\beta > \beta_0$, and any $\delta_1,\delta_2>0$ one has
\begin{equation}
\label{eq:main_res_Toda}
\bP\left( \abs{J^{(m)}\circ\phi^t_{H_F} - J^{(m)}} > \delta_1\sigma_{J^{(m)}}\right)\leq   \delta_2 C_0 \, ,
\end{equation}
for every time $t$ fulfilling 
\begin{equation}
\label{bound2}
|t| \leq \frac{ \delta_1\sqrt{\delta_2}}{\Big((\tb - 1)^2 + C_1 \beta^{-1} \Big)^{1/2}}  \beta\, .
\end{equation}
In \eqref{eq:main_res_Toda}  $\bP$ stands for the probability with respect to the Gibbs measure \eqref{eq:misura_vera}.\end{theorem}

We observe that the time scale \eqref{bound2} increases to $\beta^{\frac{3}{2}}$  for  $\tb = 1$, namely if the Toda chain 
is a fifth order approximation of the FPUT chain.

\begin{remark}
By choosing $0 < \varepsilon < \frac{1}{4}$,  $\delta_1=\beta^{-\epsilon}$ and $\delta_2=\beta^{-2\epsilon}$
the statement of the above theorem becomes:
\begin{equation}
\label{eq:main_res_Toda2}
\bP\left( \abs{J^{(m)}\circ\phi^t_{H_F} - J^{(m)}} > \frac{\sigma_{J^{(m)}}}{\beta^\varepsilon} \right)\leq   \frac{C_0}{\beta^{2\varepsilon}} \, ,
\end{equation}
for every time $t$ fulfilling 
\begin{equation}
\label{bound22}
|t| \leq \frac{ \beta^{1-2\varepsilon}}{\Big((\tb - 1)^2 + C_1 \beta^{-1} \Big)^{1/2}} \, .
\end{equation}
\end{remark}

\begin{remark}
 We observe that   our estimates  in \eqref{eq:main_res_Toda}  and \eqref{bound2} are  independent from the number of particles $N$.
Therefore we can claim that the result  of theorem~\ref{thm:goal}  holds true in the thermodynamic limit, i.e. when $\lim_{N\to \infty} \frac{\langle H_F\rangle}{N} = e > 0$
where $\langle H_F\rangle$ is the average over the Gibbs measure \eqref{eq:misura_vera} of the FPUT  Hamiltonian $H_F$.
The same observation applies to  theorem~\ref{thm:main}   and theorem~\ref{thm:main2}  below.

\end{remark}

Our Theorem \ref{thm:goal}  gives a quantitative, analytical proof of the adiabatic invariance of the Toda integrals, at least for a set of initial data of large measure. It is an interesting question whether other  integrals of motion of the Toda chain are adiabatic invariants for the FPUT chain. Natural candidates are the actions and spectral gaps.

Action-angle coordinates and the related Birkhoff coordinates  (a cartesian version of action-angle variables) were constructed analytically by Henrici and Kappeler \cite{HK1, HK2} for any  finite  $N$, and by Bambusi and one of the author  \cite{BM16} uniformly in $N$, but in a regime of specific energy going to 0 when $N$ goes to infinity (thus not the thermodynamic limit).\\
The difficulty in dealing with these other  sets of  integrals is that they are not explicit in the  physical variables $(\bp, \br)$. As a consequence, it appears very difficult  to compute their averages with respect to the Gibbs measure of the system. 
 
Despite these analytical challenges, recent numerical simulations by Goldfriend and Kurchan \cite{GK} suggest that the spectral gaps of the Toda chain are adiabatic invariants for the FPUT chain for long times also in the thermodynamic limit.

\subsection{Packets of normal modes}

Our second result concerns adiabatic invariance of some special linear combination of normal modes.
To state the result, we first  introduce the  normal modes through the discrete Hartley transform. Such  transformation, which we denote by $\cH$, is defined as 
\begin{equation}
\label{dht}
	\wh \bp  := \cH \bp,  \, \quad  \cH_{j,k} := \frac{1}{\sqrt{N}}\left(\cos\left(2\pi \frac{jk}{N}\right) + \sin\left(2\pi \frac{jk}{N}\right) \right) , \qquad  j,k = 0,\ldots, N-1
\end{equation}
and one  easily verifies that it fulfills 
\begin{equation}
\label{dht.2}
\cH^2 = \uno , \qquad \cH^\intercal  = \cH .
\end{equation}
The Hartley transform is closely related to the classical Fourier transform $\cF$, whose matrix elements are $\cF_{j,k}:= 
\frac{1}{\sqrt{N}} e^{- \im 2\pi j k/N }$, as  one has  $\cH = \Re \cF - \Im \cF$. The advantage of the Hartley transform is that it maps real variables into real variables, a fact which will be useful when calculating averages of quadratic Hamiltonians (see Section \ref{sec:lbvar}).

A consequence of \eqref{dht} is that the  change of coordinates
$$
\R^N \times \R^N \to \R^N \times \R^N, \quad  (\bp, \bq) \mapsto (\wh \bp, \wh \bq) := (\cH \bp, \cH \bq)
$$
 is a  canonical one. Due to  $\sum_j p_j =0, \, \sum_j q_j = {\mathcal L}$, one has also  $\wh p_0  =0, \, \wh q_0 = {\mathcal L}/\sqrt{N}$.  In these variables the quadratic part   $H_2$ of the  Toda  Hamiltonian \eqref{toda}, i.e. its  Taylor expansion of order two nearby the origin,   takes the form
\begin{equation} \label{quadratic_toda}
	H_2(\wh \bp, \wh \bq) := \sum_{j=1}^{N-1} \frac{\wh p_j^2 + \omega_j^2 \wh q_j^2}{2} , \qquad
	\omega_j := 2\sin\left(\pi\frac{j}{N}\right) .
\end{equation}
 We observe that \eqref{quadratic_toda} is exactly the Hamiltonian of the Harmonic Oscillator chain.  We define 
\begin{equation}
\label{he}
E_j := \frac{\wh p_j^2 + \omega_j^2 \wh q_j^2}{2} , \qquad j = 1, \ldots, N-1 \, ,
\end{equation}
 the $j^{th}$ {\em normal mode}.

To state our second result we need the following definition:
\begin{definition}[$m$-admissible vector]
\label{def:ad}
Fix $m \in \N$   and $ \wt m : = \left\lfloor\frac{m}{2}\right\rfloor $, where   $\left\lfloor\frac{m}{2}\right\rfloor $ is the integer part of $\frac{m}{2}$. 
For any $N > m$, a   vector $\bx\in \R^N$  is said to be  {\em $m$-admissible} if 
there exits   a  non zero vector  $\by=(y_0, y_1, \ldots, y_{\wt m}) \in \R^{\wt m +1}$ with $K^{-1} \leq  \sum_j |y_j| \leq K$, $K$ independent from  $N$, such that
$$ x_k= x_{N-k} = y_k, \mbox{ for  $0 \leq k \leq \wt m$   and  $x_{k} = 0$ otherwise.}    $$
\end{definition}

We are ready to state our second  result, which shows that special linear combinations of normal modes are adiabatic invariants for the FPUT dynamics for long times. Here the precise statement:
{\begin{theorem}
	\label{thm:main}
	Fix $m\in \mathbb{N}$ and  let $\boldsymbol{g}=(g_0,\dots,g_{N-1})\in\R^N$ be a $m$-admissible vector  (according to Definition \ref{def:ad}).
	Define 
\begin{equation}
\label{Phi}
\Phi := \sum_{j=0}^{N-1} \wh g_j E_j , 
\end{equation}
	 where $ \wh{\bg}$ is the discrete Hartley transform \eqref{dht}   of $\bg$,   and $E_j$ is the harmonic energy \eqref{he}. 
	Then  there exist $N_0, \beta_0, C_0, C_1>0$ (depending on $m$),  such that for any $N > N_0$, $\beta > \beta_0$, $0 < \varepsilon < \frac{1}{4}$, one has
	\begin{equation}
	\label{Hamornic_E}
	\bP\left( \abs{\Phi\circ\phi^t_{H_F} - \Phi} > \frac{\sigma_{\Phi}}{\beta^\varepsilon} \right)\leq    \frac{C_0}{\beta^{2\varepsilon}} \, , 
	\end{equation}
	for every time $t$ fulfilling 
\begin{equation}
\label{time2}
	 |t| \leq \frac{\, \beta^{1-2\varepsilon}}{\Big((\tb -1)^2 + C_2 \beta^{-1} \Big)^{1/2}}.
	 \end{equation}
\end{theorem}}

Again when  $\tb = 1$  the time scale improves by a  factor  $\beta^{\frac{1}{2}}$.

Finally we consider the Toda dynamics generated by the Hamiltonian $H_T$ in \eqref{toda}.
In this case we endow $\cM$ in \eqref{media} 
 by the Gibbs measure of $H_T$  at temperature $\beta^{-1}$, namely we put
\begin{equation}\label{eq:misura_vera_toda}
\di \mu_T :=    \frac{1}{Z_T(\beta)} \  e^{-\beta {H_T(\bp,\br)}} \,  \ \delta\left(\sum_j r_j =0\right) \ \delta\left(\sum_j p_j =0\right)  \ \di \bp \,  \di \br, 
\end{equation}
where
  as usual $Z_T(\beta)$ is the partition function which normalize the measure, namely
  \begin{equation}
Z_T(\beta) := \int_{\R^N \times \R^N}    e^{-\beta{H_T(\bp,\br)}} \ \delta\left(\sum_j r_j =0\right) \ \delta\left(\sum_j p_j =0\right)  \ \di \bp \,  \di \br . 
\end{equation}

We prove that  the quantity \eqref{Phi}, computed along the Hamiltonian flow $\phi^t_{H_T}$ of the Toda chain,  is an adiabatic invariant  {\em for all times} and for a large set of initial data:

{\begin{theorem}
	\label{thm:main2}
	Fix $m\in \mathbb{N}$; let $\bg\in \R^N$ be an $m$-admissible vector and define $\Phi$ as in \eqref{Phi}.
	Then  there exist $N_0, \beta_0, C>0$   such that for any $N > N_0$, $\beta > \beta_0$,  any $\delta_1 >0$  one has
	\begin{equation}
	\label{eq:toda_finale}
	\bP\left( \abs{\Phi\circ\phi^t_{H_T} - \Phi} > \delta_1 {\sigma_{\Phi}} \right)\leq    \frac{C}{\delta_1^2 \, \beta }  \, ,
	\end{equation}
	for all times.
\end{theorem}}

{
\begin{remark}
It is easy to verify that the functions $\Phi$ in \eqref{Phi} are linear combinations of 
\begin{equation}
\sum_{j=0}^{N-1} \cos\left( \frac{2\ell j \pi}{N}\right) \, E_j , \qquad \ell =0, \ldots, \Big\lfloor \frac{m}{2}\Big\rfloor  
\end{equation}  
(choose  $g_\ell = g_{N-\ell}=1$, $g_j = 0$ otherwise). Then, using the multi-angle trigonometric formula
$$
\cos(2nx) = (-1)^{n}T_{2n}(\sin x) , \qquad \cos(2nx) = T_{2n}(\cos x) ,
$$
where the $T_n$'s are the Chebyshev polynomial of the first kind, it follows that we can control  \eqref{sincos}.\\
{Actually these functions fall under the class considered in \cite{BCM14}.}
\end{remark}
}
Let us comment about the significance of Theorem \ref{thm:main} and Theorem \ref{thm:main2}.
   The study of the dynamics of the normal modes of FPUT goes back to the pioneering numerical simulations of Fermi, Pasta, Ulam and Tsingou \cite{Fermi}. 
   They observed that, corresponding to initial data with only the first normal mode excited, namely initial data with $E_1\neq 0$ and $E_j =0$ $\, \forall j \neq 1$, the dynamics of the normal modes develops a recurrent behavior, whereas their  time averages $\frac{1}{t}\int_0^t E_j\circ \phi^\tau_{H_F} \di \tau$  quickly relaxed to a sequence exponentially localized in $j$.
    This is what is known under the name of FPUT packet of modes. 
    
    Subsequent numerical simulations have investigated the persistence of the phenomenon for large $N$ and in different regimes of specific energies \cite{Ruffo1,Ruffo2, BERCHIALLA2, benettin-ponno11,benettin-ponno13,Onorato} (see also  \cite{BCMM15} for a survey of results about the FPUT dynamics).

Analytical results controlling packets of normal modes along the FPUT system are proven in \cite{BP06,BM16}. All these results deal with specific energies going to zero as the number of particles go to infinity, thus they do not hold in the thermodynamic limit. 
Our result controls linear combination of normal modes and holds in the thermodynamic limit.

\subsection{Ideas of the proof}
The starting point  of our analysis   is to estimate  the  probability that the time evolution  of an observable $\Phi(t)$, computed along {the Hamiltonian flow  of  $H$},   slightly deviates from its initial value.  In our application $\Phi$ is either the Toda integral of motion or a special linear combination of the harmonic energies and $H$ is either the  FPUT or Toda Hamiltonian. Quantitatively,   Chebyshev inequality gives
\begin{equation}
\label{cheb}
\bP\Big(\abs{\Phi(t) - \Phi(0)} > \lambda \sigma_{\Phi(0)} \Big)
\leq 
\frac{1}{\lambda^2} \frac{\sigma^2_{\Phi(t) - \Phi(0)}}{\sigma^2_{\Phi(0)} } \  , \qquad \forall \lambda > 0 .
\end{equation}
So our first task  is  to give  an   upper bound on the variance 
$\sigma_{\Phi(t) - \Phi(0)}$  and  a  lower bound on the variance 
$\sigma_{\Phi(0)}$.
Regarding the former  bound   we   exploit  the Carati-Maiocchi inequality \cite{CaratiMaiocchi}
\begin{equation}
\label{eq:car_maio}
\sigma_{\Phi(t) - \Phi(0)}^2 \leq \la \{ \Phi, H \}^2\ra  t^2 , \qquad \forall t \in \R ,
\end{equation}
where  $\{\Phi, H\}$, denotes the canonical Poisson bracket
{
\begin{equation}
\label{}
\{ \Phi,  H\} := (\partial_\bq \Phi)^\intercal \partial_\bp H - (\partial_\bp \Phi)^\intercal  \partial_\bq H 
\equiv
\sum_{i=0}^{N-1} \partial_{q_i} \Phi \, \partial_{p_i} H - \partial_{p_i} \Phi \, \partial_{q_i} H.
\end{equation}}
Next  we fix  $m \in \N$, consider  the $m$-th Toda  integral $J^{(m)}$, and  prove that   the quotient 
\begin{equation}
\label{quo}
\frac{\la \{ J^{(m)}, H_F\}^2\ra}{\sigma^2_{J^{(m)}} }
\end{equation}
scales appropriately in $\beta$ (as  $\beta \to \infty$) and it 
is bounded uniformly in $N$ (provided $N$ is large enough). 
It is quite  delicate to  prove that the quotient in \eqref{quo} is bounded uniformly in $N$ and for the purpose  we exploit the  rich structure of the Toda integral of motions.

This manuscript   is organized as follows.  In section 3   we study the structure of the Toda integrals.  In particular we prove that for any $m \in \N$ fixed, and $N $ sufficiently large, the $m$-th Toda  integral $J^{(m)}$  can be written as a sum  $\frac{1}{m}\sum_{j=1}^N h_j^{(m)}$ where each term depends only on at most $m$ consecutive variables, moreover  $h_j^{(m)}$ and $h_k^{(m)}$ have disjoint supports if the  distance between $j$ and $k$ is larger than $m$. Then we make the crucial observation that the quadratic part of the Toda integrals $J^{(m)}$ are quadratic forms in $\bp$ and $\bq$ generated by  symmetric circulant matrices. In section 3 we approximate the Gibbs measure with the measure  were all the variable are independent random variables.
and we calculate the error of our approximation.
In section 4 we obtain a  bound on the variance of $J^{(m)}(t)-J^{(m)}(0)$  with respect to the FPUT flow and a bound of linear combination of harmonic energies with respect to the FPUT flow and the Toda flow.
Finally in section 5 we prove our main results, namely 
Theorem~\ref{thm:goal}, Theorem~\ref{thm:main} and Theorem~\ref{thm:main2}.
We describe in the Appendices the more technical results.

\section{Structure of the Toda integrals of motion}
\label{sec:sH}

In this section we study the algebraic and the analytic properties of the Toda integrals defined in \eqref{integrals}. First we write them explicitly: 
	\begin{theorem}
	\label{LEM:STRUCT}
		For any {$1 \leq m \leq N-1$},  one has 
\begin{equation}
		\label{Jm.sum}
		J^{(m)}=  \frac{1}{m} \sum_{j=1}^N h_{j}^{(m)} \, ,
		\end{equation}		
		where $ h_{j}^{(m)}:= [L^m]_{jj}$ is given explicitly by
		\begin{equation}\label{eq:general_super_motzkin}
 h_{j}^{(m)} (\bp,\br)= \sum_{(\bn,\bk)\in \cA^{(m)}} (-1)^{|\bk|}\, \rho^{(m)}(\bn,\bk) \prod_{i = -\wt m }^{\wt m-1} e^{-{n_i} r_{j+i}} \prod_{i = -\wt m+1 }^{\wt m -1} p_{j+i}^{k_i} \, ,
		\end{equation}
where  it is understood $r_j \equiv r_{j \mod N}, \, p_j \equiv p_{j \mod N}$ and  $\cA^{(m)}$ is the set  

		\begin{equation}
		\label{cAm}
			\begin{split}
			\cA^{(m)} := \Big\{(\bn,\bk) \in \N^{\mathbb{Z}}_0 \times \N^{\mathbb{Z}}_0 \ \colon \ \ \ 
			& \sum_{i= -\wt m }^{\wt m-1} \left(2n_i + k_i\right) = m  , \\
& \forall i \geq 0, \ \ \ n_i = 0 \Rightarrow n_{i+1} = k_{i+1} = 0,  \,  
\\
& \forall i < 0, \ \ \ n_{i+1} = 0 \Rightarrow n_{i}= k_i = 0  
			\Big\}.
			\end{split}
		\end{equation}
		The quantity  $\wt m := \floor{m/2}$, $\N_0=\N\cup\{0\}$
		and $\rho^{(m)}(\bn, \bm) \in \N $ is 
		given by 
		\begin{align}
		\label{rhom}
								\rho^{(m)}(\bn,\bk) := &\binom{n_{-1} + n_0 + k_0}{k_0}\binom{n_{-1} + n_0}{n_0}
	\prod_{i=-\wt m \atop i \neq -1}^{ \wt m -1}\binom{n_i + n_{i+1} +k_{i+1} -1}{k_{i+1}}\binom{n_i + n_{i+1} -1}{n_{i+1}} \, .
\end{align}	
	\end{theorem}
	We give the proof of this theorem in Appendix \ref{app:motzkin_path}.
\begin{remark}
The structure of $J^{(N)}$ is slightly different, but we will not use it here.
	\end{remark}
	
We   now describe  some properties of the Toda integrals which we will use several times.
The   Hamiltonian density $ h_{j}^{(m)} (\bp,\br)$ depends on the set $\cA^{(m)}$ and the  coefficient $\rho^{(m)}(\bn, \bk)$ which are independent from  the index $j$. This implies that $h_j^{(m)}$ is obtained by $h_1^{(m)}$ just by shifting $1\to j$; {in \cite{giorgilli12,giorgilli15} this property was formalized with the notion of cyclic functions, we will lately recall it for completeness.}
	 
A second immediate property, as one sees inspecting  the  formulas \eqref{cAm} and \eqref{rhom}, is that 
there exists $C^{(m)} >0$ (depending only on $m$)  such that 
\begin{equation}
\label{card.Am}
|\cA^{(m)}|\leq  C^{(m)} ,  \quad \rho^{(m)}(\bn, \bk) \leq C^{(m)} , 
\end{equation}
namely the cardinality of the set $\cA^{(m)}$ and the values of the coefficients $\rho^{(m)}(\bn, \bk)$ are {\em independent of $N$}.

	The last  elementary property, which  follows from the condition $2|\bn| + |\bk| = m$ in \eqref{cAm}, is that 
\begin{equation}
\label{rem:hmj.p}
\begin{aligned}
&m \text{ even} \quad \Longrightarrow \quad h^{(m)}_j \text{ contains only even polynomials in } \bp, \\
&m \text{ odd} \quad \Longrightarrow \quad h^{(m)}_j \text{ contains only odd polynomials in } \bp.
\end{aligned}
\end{equation}

Now we describe three other important properties of the Toda integrals, which are less trivial and require some preparation. Such properties are 
\begin{itemize}
\item[$(i)$]   {\em cyclicity};
\item[$(ii)$] {\em uniformly bounded support};
\item[ $(iii)$ ] the  quadratic parts of the Toda integrals are represented by {\em circulant matrices}. 
\end{itemize}
We first define each of these  properties rigorously, and then we show that the Toda integrals enjoy them.

\paragraph{Cyclicity.}
Cyclic  functions are   characterized by being  invariant under left and right cyclic shift. 
 For any $\ell \in \Z$,  and $\bx=(x_1, x_2, \ldots, x_N)\in \mathbb{R}^N$  we define the {\em cyclic shift of order $\ell$} as the map  
\begin{equation}
\label{shift}
S_\ell \colon \R^N \to \R^N, \qquad (S_\ell x)_j := x_{(j+\ell)\mod N} . 
\end{equation}
For example $S_1$ and $S_{-1}$ are  the left respectively right shifts:
$$
S_1(x_1, x_2, \ldots, x_N) := (x_2, \ldots, x_N,  x_{1}), \qquad
S_{-1}(x_1, x_2, \ldots, x_N) := (x_N, x_1, \ldots, x_{N-1}).
$$
It is immediate to check that for  any $\ell, \ell' \in \Z$, cyclic shifts fulfills:
\begin{equation}
\label{prop:shift}
S_{\ell} \circ S_{\ell'} = S_{\ell + \ell'}, \qquad S_{\ell}^{-1} = S_{- \ell} , \qquad S_0 = \uno , \qquad
S_{\ell + N } = S_{\ell} . 
\end{equation}
Consider now a  a function $H\colon \R^N \times \R^N \to \C$;  we shall denote  $S_\ell H\colon \R^N \times \R^N \to \C$  the operator\begin{equation}
\label{cyc.func}
(S_\ell H)(\bp, \br) := H(S_\ell \bp, S_\ell \br) , \qquad \forall  (\bp, \br) \in \R^N \times \R^N.
\end{equation}
Clearly $S_\ell$ is a linear operator.
We can now define cyclic functions:
\begin{definition}[Cyclic functions]
\label{def:cyclic}
A function $H\colon \R^N \times \R^N \to \C$ is called {\em cyclic} if $S_1 H = H$.
\end{definition}
It is clear from the definition that a cyclic function fulfills $S_\ell H = H$ $\, \forall \ell \in \Z$.\\
It is easy to construct cyclic functions as follows: given a function $h \colon \R^N \times \R^N \to \C$  we define the  new function $H$ by
\begin{equation}
\label{seed}
H(\bp, \br)  := \sum_{\ell = 0}^{N-1} (S_{\ell} h)(\bp, \br) .
\end{equation}
$H$ is clearly  cyclic and  we say  that  $H$ is {\em generated } by $h$.

\paragraph{Support.}  Given a differentiable function $F \colon \R^N \times \R^N \to \C$, we define its {\em support} as the set 
\begin{equation}
\label{def:supp}
{\rm supp }\, F := \left\{ \ell \in \{ 0, \ldots, N-1\} \colon \ \ \   \frac{\partial F}{\partial p_\ell}\notequiv 0  \ \ 
\mbox{ or }
\ \   \frac{\partial F}{\partial r_\ell}\notequiv 0 \right\}
\end{equation}
and its  {\em diameter}  as 
\begin{equation}
\label{diameter}
{\rm diam} \left({\rm supp }\, F\right) := \sup_{i, j \in {\rm supp}\, F} \td(i,j) + 1 , 
\end{equation}
where $\td$ is the  {\em periodic distance}  
\begin{equation}
\label{p.dist}
\td(i,j) := \min \left( |i-j|, \ N- |i-j| \right) . 
\end{equation}
Note that $0\leq \td(i,j) \leq \floor{N/2}$.

We often use the following property:   if $f$ is a function with diameter $K \in \N$, and $K \ll N$, then
 \begin{equation}
\label{disj.supp}
\td(i,j) > K \quad \Longrightarrow \quad 
{\rm supp }\, S_j f \cap {\rm supp }\, S_i f = \emptyset ,
\end{equation}
where $S_j$ is the  shift  operator  \eqref{shift}.
With the above notation and definition we arrive to the following elementary result.
\begin{lemma}
\label{rem:diam.hj}
 Consider  the Toda integral $J^{(m)}= \frac{1}{m} \sum_{j=1}^N h_{j}^{(m)} \,$, $1 \leq m \leq N$ in \eqref{Jm.sum}. Then  $J^{(m)}$ is a cyclic function generated by  $\frac{1}{m} h^{(m)}_1$,
namely 
\begin{equation}
\label{J_cyclic}
J^{(m)}(\bp,\br)=\frac{1}{m}\sum_{j=1}^N S_{j-1} h_1^{(m)}(\bp,\br).
\end{equation}
Further,  each term $h_j^{(m)}$ has diameter at most $m$.
In particular $h_j^{(m)}$ and $h_k^{(m)}$ have disjoint supports provided $\td(j,k) > m$.
\end{lemma}

\paragraph{Circulant symmetric matrices.} 
We begin recalling the definition of circulant matrices (see e.g.  \cite[Chap. 3]{gray2006toeplitz}).
\begin{definition}[Circulant matrix]
\label{circulant}
An  $N \times N$ matrix $A$ is said to be {\em circulant} if 
there exists a vector $\ba=(a_j)_{j=0}^{N-1} \in \R^N$ such that 
$$
A_{j,k} = a_{(j-k) \mod N} .
$$
We will say that $A$ is {\em represented by the vector $ \ba$}.
\end{definition}
In particular circulant matrices have all the form
\begin{equation*}
A = {\begin{bmatrix}
 a_{0}&  a_{{N-1}}&\dots & a_{{2}}& a_{{1}}
\\
 a_{{1}}& a_{0}& a_{{N-1}}&& a_{{2}}
\\
\vdots & a_{{1}}& a_{0}&\ddots &\vdots 
\\
 a_{{N-2}}&&\ddots &\ddots & a_{{N-1}}
\\
 a_{{N-1}}& a_{{N-2}}&\dots & a_{{1}}& a_{0}
\\
\end{bmatrix}}
\end{equation*}
where each row is the right shift of the row above. \\
Moreover,    $A$ is  circulant symmetric if and only if its representing vector $\ba$ is even, i.e. one has  
\begin{equation}
\label{parity}
a_{k} =  a_{N-k}\ , \quad \forall k .
\end{equation}
One of the most remarkable property of circulant matrices is that they are all diagonalized by the discrete Fourier transform (see e.g.  \cite[Chap. 3]{gray2006toeplitz}).
We show now that circulant symmetric matrices are diagonalized by the Hartley transform:
\begin{lemma}
Let $A$ be a circulant symmetric matrix represented by the vector $\ba \in \R^N$. Then 
\begin{equation}
\label{dht.circulant}
\cH A \cH^{-1} = \sqrt{N} \, {\rm diag } \{\wh a_j \colon \ 0 \leq j \leq N-1 \},
\end{equation}
where 
$\widehat{\ba}= \cH \ba$.
	\end{lemma}
\begin{proof}
First remark that 
a circulant matrix acts on a vector $\bx \in \R^N$ as a periodic discrete convolution, 
\begin{equation}
\label{circ.conv}
A \bx = \ba \star \bx , \qquad (\ba \star \bx)_j := \sum_{k = 0}^{N-1} a_{j-k}\,  x_{k} , \qquad 0 \leq j  \leq N-1 ,
\end{equation}
where it is understood $a_\ell \equiv  a_{\ell \mod N}$.
As the Hartley transform of a discrete convolution is given by
$$
[\cH( \ba \star \bx)]_k = \frac{\sqrt{N}}{2} \Big( (\wh  a_k  + \wh a_{N-k})\wh x_k + (\wh  a_k  - \wh a_{N-k})\wh x_{N-k})\Big), 
$$
we obtain \eqref{dht.circulant}, {using that  the Hartley transform maps even vectors (see \eqref{parity}) in even vectors.}
\end{proof}

Our interest in circulant matrices comes from the following fact:
{\em quadratic cyclic functions are represented by circulant matrices}. More precisely consider a quadratic function of the form 
\begin{equation}
\label{cyc.quad}
Q(\bp, \br) = \frac{1}{2}\bp^\intercal A \bp + \frac{1}{2} \br^\intercal B \br + \bp^\intercal C \br,
\end{equation}
where  $A, B, C$ are $N \times N$ matrices. Then one has
\begin{equation}
\label{cyc.quad2}
Q \text{ is cyclic } \quad \Longleftrightarrow \quad A, B, C \text{ are circulant } .
\end{equation}
This result, which is well known (see e.g. \cite{gray2006toeplitz}), follows from the fact that $Q$ cyclic is equivalent to $A, B, C$ commuting with the left cyclic shift $S_1$, and that the set of matrices which commute with $S_1$ coincides with the set of circulant matrices.

We conclude this section collecting some properties of Toda integrals.
Denote  by $J_2^{(m)}$ the Taylor polynomial of order 2 of $J^{(m)}$ at zero; being a quadratic, symmetric, cyclic function, it is represented by circulant symmetric matrices.
We have the following lemma.
\begin{lemma}
\label{JM.STRUC}
Let us consider  the Toda integral 
\begin{equation*}
J^{(m)}(\bp,\br)=\frac{1}{m}\sum_{j=1}^N S_{j-1} h_1^{(m)}(\bp,\br).
\end{equation*}
Then  $h_1^{(m)}(\bp,\bq)$ has the following Taylor expansion at $\bp = \br = 0$:
\begin{equation}
\label{hj2}
h_1^{(m)}(\bp,\br) = \vf_0^{(m)}+ \vf_1^{(m)}(\bp,\br) + \vf_2^{(m)}(\bp,\br) +  \vf_{\geq 3}^{(m)}(\bp,\br)
\end{equation}
where   each  $\vf_k^{(m)}(\bp,\br)$  is a homogeneous polynomial of degree $k=0,1,2$ in  $\bp$  and $\br$ of diameter $m$ and coefficients independent from $N$. The reminder $ \vf_{\geq 3}^{(m)}(\bp,\br)$  takes the form
\begin{equation}
\label{reminder1}
\vf_{\geq 3}^{(m)}(\bp,\br)  := 
\sum_{(\bk, \bn) \in \cA^{(m)} \atop |k| \geq 3 } \, (-1)^{|\bk| } \rho^{(m)}(\bn, \bk) \, \bp^{\bk} \left( 1 - \bn ^\intercal \br + \frac12 (\bn^\intercal \br)^2 + \frac{(\bn ^\intercal \br)^3}{2} \int_0^1 e^{-s\bn^\intercal r} \, (1-s)^2 \, \di s \right) \,,
\end{equation}
with  $ \cA^{(m)}$ and  $ \rho^{(m)}$ defined in \eqref{cAm} and \eqref{rhom} respectively.
Moreover the   Taylor expansion   of $J^{(m)}(\bp,\br)$ at $\bp = \br = 0$ takes the form
\begin{equation}
\label{Jm.exp}
J^{(m)}(\bp, \br) = J^{(m)}_0 + J^{(m)}_2(\bp,\br)+  J^{(m)}_{\geq 3}(\bp,\br),
\end{equation}
where 
\begin{itemize}
\item[-]  {
$ J^{(m)}_0=\begin{cases}
c \in \R,  &  m \mbox{ even}  \\
0\,, &  m \mbox{ odd .}
\end{cases}$
 }
\item[-] $J^{(m)}_2(\bp,\br)$ is a  cyclic function   of the form 
\begin{equation}
\label{Jm2.struct}
J^{(m)}_2(\bp,\br) = 
 \begin{cases}
  \bp^\intercal  A^{(m)} \bp + \br^\intercal  A^{(m)} \br , &  m \mbox{ even}  \\
 \bp^\intercal  B^{(m)} \br ,  &  m \mbox{ odd}  
 \end{cases}
\end{equation}
with $A^{(m)}, B^{(m)}$  circulant, symmetric  $N \times N$  matrices; their  representing vectors  $ \ba^{(m)}$, $\bb^{(m)}$  are $m$-admissible (according to Definition \ref{def:ad})  and  
\begin{equation}
\label{ab}
a_k^{(m)} = a_{N-k}^{(m)} > 0 , \qquad
b_k^{(m)} = b_{N-k}^{(m)} > 0 , \qquad \forall 0 \leq k \leq  \wt m : =\Big\lfloor\frac{m}{2}\Big\rfloor . 
\end{equation}
\item[-]    {  The reminder  $J^{(m)}_{\geq 3}$ 
   is a cyclic function generated by $\frac{\varphi_{\geq 3}^{(m)}}{m}$}.
\end{itemize} 
\end{lemma}
The proof is postponed to   Appendix~\ref{app:struc}.
We conclude this section giving the definition of {\em $m$-admissible} functions and we prove a lemma that characterizes them in terms of $\{J_2^{(l)}\}_{l=1}^{N}$. 

\begin{definition}
\label{def:Gad}
$G_1,G_2:\R^N\times \R^N\to \R^N$ are called {\em $m$-admissible } functions of  the  first and second kind respectively if there exists a $m$-admissible vector $\boldsymbol{g}\in\R^N$   such that 
\begin{equation}
G_1 := \sum_{j,l=0}^{N-1} g_l  p_jr_{j+l} \, , \qquad G_2 := \sum_{j,l=0}^{N-1} g_l \left(p_j p_{j+l} + r_jr_{j+l}\right) \, .
\end{equation}
\end{definition}

\begin{remark}
	\label{rem:Gab.struct}
	From Definition \ref{def:Gad} and \eqref{cyc.quad2} one can deduce that both $G_1$ and $G_2$ can be represented with circulant and symmetric matrices. Indeed we have that $G_1=\bp^\intercal {\mathcal G}_1\br$ where ${(\mathcal G}_1)_{jk}=g_{(j-k) \mod N} $  and similarly for $G_2$.
\end{remark}


An immediate, but very useful,  corollary of Lemma \ref{JM.STRUC} , is the fact  that the quadratic parts of  Toda integrals are a basis of the vector space of $m$-admissible functions.
 \begin{lemma}
 \label{G.lin.comb}
 	Fix $m\in \N$ and let $G_1$ and $ G_2$  be  $m$-admissible functions of the first and second kind {
	defined by  a $m$-admissible  vector $\boldsymbol{g}\in\R^N$.  }	 Then there are two unique sequences $\{c_j\}_{j=0}^{\wt m}, \, \{d_j\}_{j=0}^{\wt m}$,  with $\max_j |c_j|, \, \max_j |d_j| $ independent from  $N$,  such that:
 	 \begin{equation}
 	 	\label{eq:linear_comb_quadratic}
 	G_1 = \sum_{l=0}^{\wt m} c_l J_{2}^{(2l+1)}, \quad 
 	G_2 = \sum_{l=0}^{\wt m} d_l J_{2}^{(2l + 2)},
 	\end{equation} 
	where $J_2^{(m)}$ is the quadratic part  \eqref{Jm2.struct}  of the Toda integrals  $J^{(m)}$ in  \eqref{Jm.sum}.
 \end{lemma}

\begin{proof}
	{We will prove the statement  just for functions of the first kind. The proof for functions of the second kind  can be obtained in a similar way.	
Let  $J^{(2l+1)}_2 =  \bp^\intercal  B^{(2l+1)} \br$ where  the  circulant  matrix  $B^{(2l+1)} $ is represented by the vector  $\bb^{(2l+1)}$ and let $G_1=\bp^\intercal {\mathcal G}_1\br$ where ${(\mathcal G}_1)_{jk}=g_{(j-k) \mod N} $. Then 
\[
{\mathcal G}_1=\sum_{l=0}^{\wt m} c_l B^{(2l+1)}\Longrightarrow g_{k}=\sum_{l=0}^{\wt m} b^{(2l+1)}_kc_l \,.
\]
	
	From Lemma \ref{JM.STRUC}  the matrix $\frak{B}=[b^{(2l+1)}_k]_{k,l=0}^{\wt{m}}$ is upper triangular and the diagonal elements are always different from $0$ (see in particular formula \eqref{ab}). This implies that the  above linear
	 system is uniquely solvable for $(c_0,\dots, c_{\wt{m}})$.	}
\end{proof}

\section{Averaging and Covariance}
In this section we collect some properties  of the  Gibbs measure $\di \mu_F$ in \eqref{eq:misura_vera}.
The first property is the invariance with respect to the shift operator. Namely for a function  $f\colon \R^N \times \R^N \to \R$;
we have that 
\begin{equation}
\label{av.shift}
\la S_j f \ra = \la f \ra , \qquad \forall j = 0, \ldots, N-1\,,
\end{equation}
which follows from the fact that $(S_j)_*\di \mu_F = \di \mu_F$.

It is in general not possible to compute exactly the average of a function with respect to   the Gibbs measure $\di \mu_F$   in \eqref{eq:misura_vera}. This is mostly due to the fact that the variables $p_0, \ldots, p_{N-1}$ and $r_0, \ldots, r_{N-1}$ are 
not independent with respect to  the measure $\di \mu_F$, being constrained by the conditions $\sum_i r_i = \sum_i p_i = 0$.

We will therefore proceed as in \cite{BCM14}, by considering a new  measure $\di \mu_{F,\theta}$   on the extended phase space according to which  all variables are independent. 
We will be able to compute averages and correlations with respect to  this measure,  and estimate the error derived by this approximation.

For any $\theta  \in \R$, we define the measure $\di \mu_{F,\theta}$ on the  extended space $\R^N\times \R^N$ by 
 \begin{equation}
\label{eq:misura_falsa}
\di \mu_{F,\theta} :=  \frac{1}{Z_{F,\theta}(\beta)} \  e^{-\beta {H_F(\bp,\br)}} \,    e^{- \theta \sum_{j=0}^{N-1} r_j}  \ \di \bp \,  \di \br, 
\end{equation}
where  we define $Z_{F,\theta}(\beta)$  as the normalizing constant of $\di \mu_{F,\theta}$.
We denote the expectation of  a function $f$ with respect to $\di \mu_{F,\theta}$ by $
\la f \ra_\theta$.
We also denote by  
$$
\norm{f}^2_\theta := \int_{\R^{2N}} |f(\bp,\br)|^2 \, \di \mu_{F,\theta} .
$$
 If $\norm{f}_\theta < \infty$ we say that $f \in L^2(\di \mu_{F,\theta})$.\\

The measure $\di \mu_{F,\theta}$ depends on the  parameter $\theta \in \R$ and we fix it in such a way that 
\begin{equation}
\label{theta.vfp}
\int_\R r \, e^{- \theta r - \beta V_F(r)} \, \di r = 0 .
\end{equation}
Following  \cite{BCM14}, it is not difficult to prove that  there exists $\beta_0 >0$ and a compact set $\cI\subset \R$  such that for any $\beta > \beta_0$, there exists  $\theta = \theta(\beta) \in \cI$ for which 
\eqref{theta.vfp} holds true. We remark that \eqref{theta.vfp} is equivalent to require  that
$\la r_j \ra_\theta = 0$  for  $\, j=0,\dots, N-1$ and as a consequence
$
\la \sum_{j=0}^{N-1} r_j \ra_\theta = 0 .$
We observe that  $\la \sum_{j=0}^{N-1} r_j \ra = 0$ with respect to the  measure $\di \mu_F$. 

The main reason for introducing  the measure $\di \mu_{F, \theta}$ is that it approximates  averages  with respect to $\di \mu_F$ as the following result shows.
\begin{lemma}
\label{magic}
Fix $\wt\beta >0$  and let $f \colon \R^N\times \R^N \to \R$ have support of size $K$ (according to Definition \ref{def:supp}) and  finite second order moment with respect to $\di \mu_{F,\theta}$, uniformly for all $\beta > \wt\beta$.
Then there exist   positive constants $C, N_0$ and $\beta_0 $ such that
for all $N > N_0$, $\beta > max\{\beta_0,\tilde{\beta}\}$ one has  
\begin{equation}
\label{eq:magic}
\abs{\la f \ra - \la f \ra_\theta }\leq  C \frac{K}{N} \sqrt{\la f^2\ra_{\theta} - \la f \ra_\theta^2}\, .
\end{equation}
\end{lemma}
The above  lemma is   an extension to the periodic case of a result from  \cite{BCM14}, and we shall prove it in Appendix  \ref{app:meas.app}.
As an example of applications of Lemma \ref{magic}, we  give a bound to  correlations  functions.
\begin{lemma}
\label{cov}
Fix $K \in \N$. Let $f, g \colon \R^{N}\times \R^N \to \C$ such that :
\begin{itemize}
\item[1.] $f, g$ and $ fg \in L^2(\di \mu_{F,\theta})$,
\item[2.] the supports of $f$ and $g$ have size at most $K \in \N$.
\end{itemize}
Then there exist $C, N_0, \beta_0 >0$ such that for all $N > N_0$, $\beta > \beta_0$ 
\begin{equation}
 \label{cov1}
{\abs{ \la f g \ra - \la f \ra \la g \ra } \leq 2\norm{f}_\theta \, \norm{g}_\theta+ \frac{C K}{N} \Big(\norm{f}_\theta \, \norm{g}_\theta + \norm{fg}_\theta\Big) . }
 \end{equation} 
 Moreover, if $f$ and $g$ have disjoint supports, then
 \begin{equation}
 \label{cov2}
\abs{ \la f g \ra - \la f \ra \la g \ra } \leq  \frac{C K}{N} \Big(\norm{f}_\theta \, \norm{g}_\theta + \norm{fg}_\theta \Big) .
\end{equation}
\end{lemma}
\begin{proof}
We substitute  the measure  $\di \mu_F$ with $\di \mu_{F,\theta}$ and  then we  control the error by  using Lemma \ref{magic}. With this idea, we write
\begin{align}
\label{pezzo1}
\la f g \ra - \la f \ra \la g \ra  = & 
 \la f g  \ra - \la f g  \ra_{\theta} \\
 \label{pezzo2}
 & + \la f g  \ra_\theta - \la f \ra_\theta \la g \ra_\theta \\
 \label{pezzo3}
 & + \la f \ra_\theta \la g \ra_\theta - \la f \ra \la g \ra \, ,
\end{align}
and estimate the different terms.
We will  often use  the inequality 
\begin{equation}
\label{L1}
\abs{\la f \ra_\theta} \leq \norm{f}_\theta \, ,
\end{equation}
 valid for any function $f \in L^2(\di \mu_{F,\theta})$.
 
\noindent
{\sc Estimate of \eqref{pezzo1}:}  By Lemma \ref{magic}, and the assumption that $f g$ depends on at most  $2K$ variables,
\begin{align}
\notag
\abs{\la fg \ra - \la fg \ra_{\theta}} 
&\leq 
C \frac{2K}{N} \sqrt{\la (fg)^2\ra_\theta - \la fg\ra_\theta^2}
 \leq \frac{C' K }{N} \norm{fg}_{\theta} \, .
\end{align}

\noindent
{\sc Estimate of \eqref{pezzo2}:} By Cauchy-Schwartz  and \eqref{L1} we have
\begin{align}
  \label{pezzo2.est}
\abs{\la  fg  \ra_\theta - \la f \ra_\theta \la g \ra_\theta}
\leq  2 \norm{f}_\theta \norm{g}_\theta \, .
 \end{align} 

\noindent
{\sc Estimate of \eqref{pezzo3}:} We decompose further
\begin{align*}
 \la f \ra_\theta \la g \ra_\theta - \la f \ra \la g \ra
  = & \la g \ra_\theta \left( \la f\ra_\theta - \la f \ra\right) 
 + \left(\la g \ra_\theta - \la g \ra \right)\la f \ra_\theta  + \left(\la g \ra_\theta - \la g \ra \right)\left( \la f \ra - \la f\ra_\theta \right) \, ,
  \end{align*}
  again by Lemma \ref{magic} and \eqref{L1} we obtain
  \begin{align}
  \notag
  \abs{  \la f \ra_\theta \la g \ra_\theta - \la f \ra \la g \ra}
  & \leq 
   C\frac{K}{N} \norm{g}_\theta \norm{f}_\theta \, .
  \end{align}
{ Combining  the  three  bounds    above and redefining $C=\mbox{max}\{C,C'\}$ one obtains   \eqref{cov1}. }
  To prove \eqref{cov2} it is sufficient to observe that if $f$ and $g$ have disjoint supports, then $\la fg\ra_\theta = \la f\ra_\theta \la g \ra_\theta$ and consequently \eqref{pezzo2} is equal to zero. 
\end{proof}

In order to make Lemma \ref{cov} effective we need to show how to compute averages according to  the  measure \eqref{eq:misura_falsa}. 
\begin{lemma}
\label{LEM:ORDINE_BETA}
There exists $\beta_0 >0$ such that for any   $\beta > \beta_0$, the following holds true.  For  any fixed  multi-index $\bk,\bl,\bn,\bs \in \N_0^{N}$ and $d,d' \in \{ 0,1,2\}$,  there are two constants $C_{\bk,\bl}^{(1)} \in \R$ and $C_{\bk,\bl}^{(2)}>0$ such that 
\begin{equation*}
\frac{C^{(1)}_{\bk,\bl}}{\beta^{\frac{|\bk| + |\bl| }{2}}} \leq \la \bp^\bk \, \br^\bl \,  \left(\int_{0}^1  e^{- \xi  \bn^\intercal \br  }(1-\xi)^2\di \xi\right)^d  \left(\int_{0}^1  e^{- \xi  \bs^\intercal \br  }(1-\xi)^3\di \xi\right)^{d'} \ra_\theta \leq \frac{C^{(2)}_{\bk,\bl}}{\beta^{\frac{|\bk| + |\bl | }{2}}}
\end{equation*}
where $\bp^\bk=\prod_{j=1}^Np_j^{k_j}$ and $ \br^\bl=\prod_{j=1}^N r_j^{l_j} $.
Moreover:
\begin{itemize}
	\item[(i)] if $k_i$ is odd for some $i$ then $C^{(1)}_{\bk,\bl} = C^{(2)}_{\bk,\bl} = 0$;
	\item[(ii)] if $k_i, l_i$ are even for all $i$ then $C^{(1)}_{\bk,\bl} >0$.
\end{itemize}
\end{lemma}
The lemma is proved in Appendix \ref{app:process}.

\begin{remark}
Actually   all the results of this section hold true (with different constants)  also when we endow $\cM$ with the Gibbs measure of the Toda chain in \eqref{eq:misura_vera_toda} and we use as approximating measure 
 \begin{equation}
\label{eq:misura_falsaT}
\di \mu_{T,\theta} :=  \frac{1}{Z_{T,\theta}(\beta)} \  e^{-\beta {H_T(\bp,\br)}} \,    e^{- \theta \sum_{j=0}^{N-1} r_j}  \ \di \bp \,  \di \br;
\end{equation}
here $\theta$ is selected in such a way that
\begin{equation}
\label{theta.vt}
\int_\R r \, e^{- \theta r - \beta V_T(r)} \, \di r = 0 .
\end{equation}
We show in  Appendix  \ref{app:process} that it is always possible to choose $\theta$ to fulfill \eqref{theta.vt} (see Lemma \ref{lem:mom_r}) and we also prove Lemma \ref{LEM:ORDINE_BETA} for Toda. In Appendix  \ref{app:meas.app}  we prove Lemma \ref{magic} for the Toda  chain.
\end{remark}

%
%
%
%
%
%

\section{Bounds on the Variance}
In this section we prove  upper and lower bounds on the variance of the  quantities  relevant to prove our main theorems.

\subsection{Upper bounds on the variance of $J^{(m)}$ along the flow of FPUT}
In this subsection we only consider the case $\cM$  endowed by the FPUT Gibbs measure. 
{We denote by $J^{(m)}(t) := J^{(m)}\circ \phi^t_{H_F}$ the Toda integral computed along the Hamiltonian flow $\phi^t_{H_F}$ of the FPUT Hamiltonian.}
The aim  is to prove the following result:
\begin{proposition}
\label{lem:numeden}
Fix $m \in \N$.  There exist $N_0, \beta_0, C_0, C_1>0$  such that for any $N > N_0$, $\beta > \beta_0$, one has
\begin{align}
\label{num}
& \sigma_{J^{(m)}(t) - J^{(m)}(0)}^2 \leq C_0  N \left(\frac{(\tb - 1)^2}{\beta^4} + \frac{C_1}{\beta^{5}}\right) t^2  , \qquad \forall t \in \R .
\end{align}
\end{proposition}
\begin{proof}
As   explained in the introduction, applying formula 
 \eqref{eq:car_maio}  we get 
\begin{equation}
\label{num2}
\sigma_{J^{(m)}(t) - J^{(m)}(0)}^2 \leq \la \{ J^{(m)}, H_F\}^2\ra  t^2 , \qquad \forall t \in \R.
\end{equation} 
 Therefore we need  to  bound 
$\la \{J^{(m)}, H_F\}^2 \ra$.
For the purpose we rewrite this term in a more  convenient form.
Since
$\la \cdot \ra$ is an invariant measure with respect to the Hamiltonian flow of $H_F$,  one has 
\begin{equation}
\label{J_mH_F}
\la \{ J^{(m)}, H_F \} \ra = 0 .
\end{equation}
Furthermore, since  $J^{(m)}$ is an integral of motion of the Toda Hamiltonian $H_T$, we have
		\begin{equation}
		\label{J_mH_T}
			\left\{J^{(m)}, H_T\right\} = 0  .
		\end{equation}
We apply identities \eqref{J_mH_F} and \eqref{J_mH_T} to   write
		\begin{align}
	\la\left\{J^{(m)}, H_F\right\}^2\ra 
\label{Jm2}
	= \la\left\{J^{(m)},  H_F - H_T\right\}^2 \ra - \la \{ J^{(m)}, H_F- H_T \} \ra^2 \, .
	\end{align}
The  above expression enables us to exploit the fact that the FPUT system is a  { fourth order} perturbation  of the Toda chain.
To proceed with the proof we need the following technical result.
\begin{lemma}
\label{lem:J_mR}
One  has 
\begin{equation}
\label{J_mR}
\{J^{(m)}, H_F-H_T \}= \sum_{j=1}^N H_j^{(m)} \, ,
\end{equation} 
where the functions $H_j^{(m)}$ fulfill
\begin{itemize}
\item[(i)]$H_j^{(m)} = S_{j-1} H_1^{(m)}$ $\ \forall j $, moreover  the diameter of the  support of $H_j^{(m)}$  is  at most $m$;
\item[(ii)] there exist $N_0, \beta_0, C, C'>0$ such that for any $N > N_0$, $\beta > \beta_0$,  any $i, j = 1, \ldots, N$, the following estimates hold true:
 \begin{align}
  \label{avgHi}
\norm{H_j^{(m)}}_\theta \leq C\left( \frac{(\tb-1)^2}{\beta^4} + \frac{C'}{\beta^5} \right)^{1/2},  \qquad 
 \norm{ H_i^{(m)} H_j^{(m)} }_\theta  \leq  C\left(\frac{\left(\tb-1\right)^4}{\beta^8} + \frac{C'}{\beta^{10}}\right)^{1/2} .
\end{align}
\end{itemize}
\end{lemma}
The proof of the lemma is postponed at the end of the subsection.\\
  We are now ready to finish the proof of  Proposition~\ref{lem:numeden}.
Substituting \eqref{J_mR} in \eqref{Jm2} we obtain
\begin{align}
\label{JmHF2}
	\la\left\{J^{(m)}, H_F\right\}^2\ra &  = 
		\sum_{j,i=1}^{N}
\left[\la H_i^{(m)} H_j^{(m)} \ra - \la H_i^{(m)} \ra \la H_j^{(m)} \ra\right] .
\end{align}
Therefore estimating $\la\left\{J^{(m)}, H_F\right\}^2\ra$ is equivalent to estimate the  correlations between  $H_i^{(m)}$ and $H_j^{(m)}$. 
Exploiting Lemma \ref{cov} and observing that if $\td(i,j) > m$ then $H_i^{(m)}$ and $H_j^{(m)}$ have disjoint supports (see  Lemma \ref{lem:J_mR} $(i)$ and \eqref{disj.supp}), we get that  there are 
positive constants that for convenience we still call  $C$ and  $C'$, such that  $\forall N, \beta$ large enough
\begin{align}
\label{est1}
&\abs{\la H_i^{(m)} H_j^{(m)}\ra - \la H_i^{(m)} \ra \la H_j^{(m)} \ra} \leq C \left(\frac{(\tb- 1)^2}{\beta^4}  + \frac{C'}{\beta^5} \right) , \qquad \forall i, j , \\
\label{est2}
&\abs{\la H_i^{(m)} H_j^{(m)}\ra - \la H_i^{(m)} \ra \la H_j^{(m)} \ra}  \leq \frac{C}{N} \left(\frac{(\tb - 1)^2}{\beta^4}  + \frac{C'}{\beta^5} \right) , \qquad \forall i, j \colon \td(i, j) > m\,.
\end{align}
From \eqref{JmHF2} we split the sum in two terms:
\begin{align*}
	\la\left\{J^{(m)}, H_F\right\}^2\ra &  = 
  \sum_{\td(i,j) \leq m}\left[
\la H_i^{(m)} H_j^{(m)} \ra - \la H_i^{(m)} \ra \la H_j^{(m)} \ra \right]\\
& \quad + \sum_{\td(i,j)>m}\left[
\la H_i^{(m)} H_j^{(m)} \ra - \la H_i^{(m)} \ra \la H_j^{(m)} \ra\right] \, .
\end{align*}
We now apply estimates \eqref{est1}, \eqref{est2} to get 
\begin{align}
\nonumber
\la\left\{J^{(m)}, H_F\right\}^2\ra &  \leq N  C\left(\frac{(\tb - 1)^2}{\beta^4}  + \frac{C'}{\beta^5} \right)  + N^2 \frac{\wt C}{N} \left(\frac{(\tb-1)^2}{\beta^4}  + \frac{C'}{\beta^5} \right) \\&\leq N C_1 \left(\frac{(\tb - 1)^2}{\beta^4}  + \frac{C_2}{\beta^5} \right) \,
\end{align}
for some positive constants $C_1$ and $C_2$.
\end{proof}

\subsubsection{Proof of Lemma \ref{lem:J_mR}}

We start by  writing  the Poisson bracket $\{J^{(m)}, H_F- H_T\}$ in an explicit form.
First we observe that  for any $1 \leq m  < N$ one has from \eqref{integrals}
\begin{equation}
\label{Jm_der_pj}
\frac{\partial J^{(m)}}{\partial p_{j-1}} = \frac{1}{m} \frac{\partial \text{Tr}\left(L^m\right)}{\partial p_{j-1}} =  \text{Tr}\left( L^{m-1}\frac{\partial L}{\partial p_{j-1}}\right) = - [L^{m-1}]_{j,j}  = -h_{j}^{(m-1)} , 
\end{equation}
 for all $j=1, \ldots, N$.  {In the above relation $ h_{j}^{(m-1)}$ is the generating function of the $m-1$ Toda integral defined in \eqref{eq:general_super_motzkin}.}

Next we observe that 
\begin{equation}
\label{def:R}
 H_F(\bp,\bq) -   H_T(\bp,\bq)  =   \sum_{j=0}^{N-1} R(q_{j+1}-q_j), \qquad R(x) :=  \frac{x^2}{2}- \frac{x^3}{6} + \tb \frac{x^4}{24} - (e^{-x} -1 + x) .
\end{equation}
This implies also that
\begin{equation}
\label{eq:Pbra}
\begin{split}
\left\{J^{(m)}, H_F - H_T\right\} &=  \sum_{j=1}^{N} h_{j}^{(m-1)}(\bp,\br) \,  \left( R'(r_{j-2}) - R'(r_{j-1}) \right)\\
&= \sum_{j=1}^{N}( h_{j}^{(m-1)}(\bp,\br)-h_{j}^{(m-1)}(\mathbf{0},\mathbf{0})) \,  \left( R'(r_{j-2}) - R'(r_{j-1}) \right)
\end{split}\end{equation}
where, to obtain the second identity, we are using  that  $h_{j}^{(m-1)}(\mathbf{0},\mathbf{0})$ is  by \eqref{J_cyclic} and \eqref{hj2}  a constant independent from $j$ and the second term in the last relation is a telescopic sum.
 Define
 \begin{equation}
 \label{def:Hj}
 H_j^{(m)}:=  \left( h_{j}^{(m-1)}(\bp,\br) -  h_{j}^{(m-1)}(\mathbf{0},\mathbf{0})\right) \left( R'(r_{j-2}) - R'(r_{j-1}) \right), \quad j=1,\ldots,N ;
 \end{equation}
then   item (i)  of  Lemma~\ref{lem:J_mR}  follows  because clearly  $H_j^{(m)}=S_{j-1}H_1^{(m)}$. Furthermore, since $h_j^{(m-1)}$ has diameter bounded by $m-1$, the same property applies to  $H_j^{(m)}$.

%
%

To prove item $(ii)$ we start by expanding $R'(r_{j-1}) - R'(r_j)$ in Taylor series with integral remainder. Since
$$
R'(x) = \frac{(\tb - 1)}{6} x^3 + \frac{x^4}{6}\int_{0}^1 e^{-\xi x}(1-\xi)^3\di \xi \, ,
$$
we get that 
{
\begin{align}
\label{R'.exp}
& R'(r_{j-2}) - R'(r_{j-1})    = \frac{(\tb - 1)}{6}S_{j-1}\psi_3(\br) + \frac{1}{6}S_{j-1}\psi_4(\br)\, , 
\end{align}
where explicitly
\begin{align}
\label{psi3}
&\psi_3(\br) := r_{N-1}^3 - r_{0}^3 \, ,\\
\label{psi4}
&\psi_4(\br) := r_{N-1}^4\int_{0}^1 e^{-\xi r_{N-1}}(1-\xi)^3\di \xi -  r_{0}^4\int_{0}^1 e^{-\xi r_{0}}(1-\xi)^3\di \xi \,.
\end{align}
Combining \eqref{hj2} with \eqref{R'.exp} we  rewrite $H_j^{(m)}$ in \eqref{def:Hj} in the form
}
\begin{equation*}
	H^{(m)}_j  = \frac{S_{j-1}}{6}\left((\vf^{(m)}_{1}+\vf^{(m)}_{2}+\vf^{(m)}_{\geq3})\Big( (\tb - 1) \psi_3  + \psi_4 \Big) \right) \,,
		\end{equation*}	
		where $\vf^{(m)}_{j}$, $ j=0,1,2$, are defined in \eqref{hj2}.
	Thus the squared $L^2$ norm of $H_j$  is given by (we suppress the superscript to simplify the notation)
	\begin{align}
	\label{Hj.norm1}
	\norm{H_j}_\theta^2 & = \frac{1}{36}(\tb -1 )^2\Big(\sum_{\ell, \ell' = 1}^{2} \,  \la \psi_3^2 \, \vf_{\ell} \,  \vf_{\ell'}\ra_{\theta} +\,  \la \psi_3^2 \vf_{\geq 3}\left( \vf_{\geq 3} + 2\vf_{1} + 2\vf_{2}\right)\ra_\theta\Big) \\ 
	\label{Hj.norm2}
	&+\frac{\tb -  1 }{18} \Big(\sum_{\ell, \ell' = 1}^{2}\,  \la \psi_3 \psi_4  \, \vf_{\ell} \,  \vf_{\ell'} \ra_\theta  + \la \psi_3 \psi_4  \vf_{\geq 3}\left( \vf_{\geq 3} + 2\vf_{1} + 2\vf_{2}\right) \ra_\theta\Big)\\
	\label{Hj.norm3}
	& + \frac{1}{36}\sum_{\ell, \ell' = 1}^{2}\la \psi_4^2 \, \vf_{\ell} \,  \vf_{\ell'} \ra_\theta +\frac{1}{36}\la \psi_4^2 \, \vf_{\geq 3}\left( \vf_{\geq 3} + 2\vf_{1} + 2\vf_{2}\right)\ra_\theta \, .
\end{align}	
Consider now the terms in \eqref{Hj.norm1}; 
by \eqref{Jm2.struct},\eqref{reminder1}  and \eqref{psi3},  we know that each element is  a linear combinations of functions of the form
\begin{equation}
\label{Hj.norm4}
\bp^\bk \, \br^\bl \,  \left(\int_{0}^1  e^{- \xi \bn^\intercal\br  }(1-\xi)^2\di \xi\right)^d  \left(\int_{0}^1  e^{- \xi \bs^\intercal\br  }(1-\xi)^3\di \xi\right)^{d'}\, ,
\end{equation}
with $|\bk| + |\bl | \geq 6 + \ell + \ell' \geq  8$, $d,d'\in \{0,1,2\}$. The number of these functions and their  coefficients are independent from $N$ (see Lemma \ref{JM.STRUC}). By Lemma \ref{LEM:ORDINE_BETA} it follows that there exists a constant $C>0$,  depending only on $m$, such that
\begin{equation}
\abs{\mbox{r.h.s. of \eqref{Hj.norm1}}} \leq  C \, (\tb - 1)^2 \,  \beta^{-4} . 
\label{Hj.norm5}
\end{equation}
Analogously, line \eqref{Hj.norm2} is a linear combination of functions of the form \eqref{Hj.norm4} with $|\bk| + |\bl | \geq 9$, $d,d'\in \{0,1,2\}$. 
Applying Lemma \ref{LEM:ORDINE_BETA} we get the estimate
\begin{equation}
\abs{\eqref{Hj.norm2}} \leq C' \, |\tb -1| \,  \beta^{-9/2} \, 
\label{Hj.norm6}
\end{equation}
for some constant $C'>0$. In a similar way
 the expression   \eqref{Hj.norm3} is a linear combination of functions of the form \eqref{Hj.norm4} with $|\bk| + |\bl | \geq 10$, $d,d'\in \{0,1,2\}$. 
Applying Lemma \ref{LEM:ORDINE_BETA} we get the estimate
\begin{equation}
\abs{\eqref{Hj.norm3}} \leq C'' \,  \beta^{-5} \, ,
\label{Hj.norm7}
\end{equation}
for some constant $C''>0$.
Combining  \eqref{Hj.norm5},\eqref{Hj.norm6} and \eqref{Hj.norm7} we obtain  estimate  \eqref{avgHi} for $\norm{H_j}_\theta$.
 The estimate  for $\norm{ H_i^{(m)} H_j^{(m)} }_\theta $ can be  proved in an analogous way.

\qed

\subsection{Lower  bounds on the variance of $m$-admissible functions}
\label{sec:lbvar}
From now on  we consider $\cM$ endowed with either the  FPUT or the  Toda Gibbs measure;  the following result holds in both cases.
\begin{proposition}
\label{prop:ICTP}
Fix $m \in \N$,   let $G$ be an  $m$-admissible function of the first or  second kind (see Definition \ref{def:Gad}). There exist $N_0, \beta_0, C >0$ such that for any $N > N_0$, $\beta > \beta_0$, one has
\begin{align}
\label{eq:var_G1}
&\sigma^2_{G} =\la G^2\ra-\la G\ra^2 \geq  C\frac{N}{\beta^2} . 
\end{align}
\end{proposition}
\begin{proof}
We first prove  \eqref{eq:var_G1} when $G = G_1= \bp^\intercal \cG_1 \br$ where $\cG_1$ is a circulant, symmetric matrix represented by the $m$-admissible vector $\ba \in \R^N$.    We now make the change of coordinates $(\bp, \br) = (\cH \wh \bp, \cH \wh \br)$ which  {diagonalizes  the matrix $\cG_1$}  (see \eqref{dht.circulant}), getting 
\begin{equation*}
G_1(\wh \bp,\wh \br) = \sqrt{N}\sum_{j=0}^{N-1}\wh g_j \wh p_j\wh r_j.
\end{equation*}
So we have just to compute
\begin{align}
\notag
\sigma_{G_1}^2  & = N \la \sum_{i,j = 0}^{N-1} \wh g_j\wh g_i \wh p_j \wh p_i \wh r_i \wh r_j \ra - N\left(\la \sum_{j=0}^{N-1}\wh g_j \wh p_j\wh r_j\ra\right)^2 \\
\label{eq:first_step_G1}
& =    N  \sum_{i,j = 0}^{N-1} \wh g_j\wh g_i \la \wh p_j \wh p_i\ra  \la \wh r_i \wh r_j \ra - N \left(\sum_{j=0}^{N-1}\wh g_j \la \wh p_j \ra \la \wh r_j\ra\right)^2 \, ,
\end{align}
where we used  that $\wh p_k, \wh r_j$ are  random variables independent from each other.

We notice that $\wh p_1, \wh p_2,\ldots, \wh p_{N-1}$ are i.i.d. Gaussian random variable with variance $\beta^{-1}$,   $\wh p_0 = 0$ (see \eqref{m_qp}), so  that  we have $\la \wh p_j \ra = 0$ and   $\la \wh p_j\wh p_i\ra = \frac{\delta_{i,j}}{\beta} \, \, i,j = 1,\ldots, N-1$ (remark that this holds true both for the  FPUT and Toda's potentials as the $p$-variables have the same distributions).\\
As a consequence,   \eqref{eq:first_step_G1} becomes:
\begin{equation}
\label{eq:third_step_G1}
\sigma^2_{G_1} = \frac{N}{\beta}  \sum_{j = 1}^{N-1} \wh{g}_j^2 \la \wh r_j^2\ra  = \frac{1}{\beta} 
\la \wh \br^\intercal \cH \cG_1^2 \cH \wh  \br \ra =\frac{1}{\beta}  \la \br^\intercal \cG_1^2  \br \ra . 
\end{equation}
Since  $\cG_1$ is circulant  symmetric matrix  so is $\cG_1^2$ and its representing vector is  $\bd := \bg \star \bg $. \\
Next we  remark that   the identity  $\la \left(\sum_{j=0}^{N-1} r_j\right)^2\ra = 0$  implies 
\begin{equation*}
\la r_jr_i\ra = - \frac{1}{N-1}\la r_0^2 \ra , \quad \forall i\ne j \, .
\end{equation*}
Applying this property to \eqref{eq:third_step_G1} we get
\begin{align}
\nonumber
\sigma^2_{G_1} & = \frac{1}{\beta}
 \sum_{j,l=0}^{N-1} 
 d_l \, \la  r_j r_{j+l }\ra  = \frac{N}{\beta} \la r_0^2\ra d_0 + \frac{1}{\beta}\sum_{j, l \atop l \neq 0}^{N-1} d_l \la r_j r_{j+l} \ra 
\\
\label{sigma_est}
& = \frac{1}{\beta} \la r_0^2\ra \left(Nd_0 - \frac{N}{N-1}\sum_{l \neq 0}^{N}d_l \right) \,.
\end{align}
By Lemmas \ref{magic} and \ref{LEM:ORDINE_BETA} we have that, for $N$ sufficiently large, $\la r_0^2 \ra \geq c \beta^{-1}$. 
Finally, since the vectors $\bg,\bd$ are $m$-admissible and $2m$-admissible respectively we have that
\begin{equation}
\label{d0}
d_0 = (\bg \star \bg)_0 = \sum_{j=0}^{\wt m} g_j^2 \geq c_m , \qquad 
\sum_{l \neq 0}^{N-1}d_l  = \sum_{l \neq 0}^{2 \wt m}d_l  \leq C_m,
\end{equation}
for some constants $c_m>0$ and $C_m>0$. Plugging \eqref{d0} into \eqref{sigma_est}
we obtain \eqref{eq:var_G1} for the  case of 
 $m$-admissible functions of the first kind.

For the  case  of admissible functions of the second kind, one has  $G_2 = \bp^\intercal \cG_2 \bp + \br^\intercal \cG_2 \br$ with  $\cG_2$  circulant, symmetric and  represented by an  $m$-admissible vector. 
Since  $\bp$ and $ \br$ are independent   random variables one gets
$$\sigma_{G_2} =  \sigma_{\bp^\intercal \cG_2 \bp + \br^\intercal \cG_2 \br} = \sigma_{\bp^\intercal \cG_2 \bp}  +\sigma_{ \br^\intercal \cG_2 \br}\geq \sigma_{\bp^\intercal \cG_2 \bp} . $$
Then arguing as in the previous case one gets \eqref{eq:var_G1}.
\end{proof}
By applying Proposition \ref{prop:ICTP} to  the quantity $J^{(m)}_2$  that  is an $m$-admissible function of the first or second kind depending on the parity of $m$,
  we obtain the following result.

\begin{corollary}
\label{corollary_quad}
The quadratic part $J_2^{(m)}$ of the  Taylor expansion of the  Toda integral $J^{(m)}$  near $(\bp, \br ) = (0, 0) $  satisfies
\begin{equation}
	\label{rem:J2_admissible}
	\sigma^2_{J^{(m)}_2}\geq C \frac{N}{\beta^2}\,,
\end{equation}
for some constant $C>0$.
\end{corollary}
In a similar way we obtain a lower bound on the reminder  $J^{(m)}_{\geq 3}$ of the  Taylor expansion of the Toda integral $J^{(m)}$  near $\bp=0$ and $\br=0$.
\begin{lemma}
\label{cor:stima_coda}
Fix $m \in \N$. There exist $N_0, \beta_0, C >0$  such that for any $N > N_0$, $\beta > \beta_0$, one has
\begin{equation}
\label{est:Jm3}
\sigma^2_{J^{(m)}_{\geq 3}} \leq C \frac{ N}{\beta^3} . 
\end{equation}
\end{lemma}
\begin{proof}
Recall from Lemma \ref{JM.STRUC} that $J^{(m)}_{\geq 3}$ is a cyclic function generated by   $\wt h_1^{(m)}:=\frac{1}{m} \vf^{(m)}_{\geq 3}$. Thus, denoting  $h^{(m)}_j := S_{j-1} \wt h^{(m)}_1$, we have 
$
J^{(m)}_{\geq 3} = \sum_{j = 1}^N \wt h^{(m)}_j $
and its variance is given by
\begin{equation}
\label{var.Jm3}
\sigma^2_{J_m^{\geq 3}} = \sum_{i,j=1}^N \la \wt h^{(m)}_i\wt h^{(m)}_j\ra - \la \wt h^{(m)}_i \ra \la \wt h^{(m)}_j \ra . 
\end{equation}  
We  can bound the correlations in \eqref{var.Jm3}  exploiting  Lemma \ref{cov}, provide we estimate first  the $L^2(\di \mu_{F,\theta})$ and $L^2(\di \mu_{T,\theta})$ norms of
$\wt h^{(m)}_i $ and $\wt h^{(m)}_i  \wt h^{(m)}_j$. 
Proceeding with the same arguments as in Lemma  \ref{lem:J_mR}, one proves that there exists $\tilde{C}>0$ such that  for any $N > N_0$, $\beta > \beta_0$, 
\begin{equation}
  \label{avgHi1}
\norm{\wt h^{(m)}_i}_\theta \leq \tilde{C} \beta^{-3/2}, \qquad
\norm{\wt h^{(m)}_i \, \wt h^{(m)}_j}_\theta \leq \tilde{C} \beta^{-3}.
\end{equation}
By Lemma \ref{JM.STRUC}, the function  $\wt h_1^{(m)}$ has diameter at most $m$, so in particular  if  $\td(i,j) > m$, the functions  $\wt h_i^{(m)}$ and $\wt h_j^{(m)}$ have disjoint supports (recall  \eqref{disj.supp}).\\
We are now in position to apply Lemma \ref{cov} and obtain 
\begin{align}
\label{est11}
&\abs{\la \wt h^{(m)}_i \wt h^{(m)}_j\ra - \la \wt h^{(m)}_i \ra \la \wt h^{(m)}_j \ra} \leq \frac{C'}{\beta^3}   , \qquad \forall i,j \\
\label{est22}
&\abs{\la \wt h^{(m)}_i\wt h^{(m)}_j\ra - \la \wt h^{(m)}_i \ra \la \wt h^{(m)}_j \ra}  \leq \frac{C'}{N\beta^3} , \qquad \forall i, j \colon \td(i,j) > m,
\end{align}
for some constant $C'>0$.
Thus we split the variance in \eqref{var.Jm3} in two parts
\begin{equation*}
\sigma^2_{J^{(m)}_{\geq 3}} = \sum_{\td(i,j)\leq m} \la \wt h^{(m)}_i\wt h^{(m)}_j\ra - \la \wt h^{(m)}_i \ra \la \wt h^{(m)}_j \ra + \sum_{\td(i,j) > m} \la \wt h^{(m)}_i\wt h^{(m)}_j\ra - \la \wt h^{(m)}_i \ra \la \wt h^{(m)}_j \ra 
\end{equation*}  
and apply estimates \eqref{est11}, \eqref{est22} to get  \eqref{est:Jm3}.
\end{proof}
Combining Corollary~\ref{corollary_quad} and Lemma~\ref{cor:stima_coda}  we arrive to the following crucial proposition.
\begin{proposition}
	\label{prop:sJm}
	Fix $m\in \N$.  There exist $N_0, \beta_0, C >0$ such that for any $N > N_0$, $\beta > \beta_0$, one has
	\begin{equation}
	\label{est:sJm}
	\sigma^2_{J^{(m)}}\geq C \frac{N}{\beta^2} .
	\end{equation}
\end{proposition} 
\begin{proof}
By Lemma \ref{JM.STRUC}, we write  $J^{(m)} = J^{(m)}_0 + J^{(m)}_2+  J^{(m)}_{\geq 3}$ with $J^{(m)}_0$  constant.
By Corollary~\ref{corollary_quad} and Lemma~\ref{cor:stima_coda}  we deduce   that for $N$ and $\beta$ large enough, 
$$
\sigma_{J^{(m)}} = \sigma_{J^{(m)}_2 + J^{(m)}_{\geq 3}}  \geq \sigma_{J^{(m)}_2 }  - \sigma_{ J^{(m)}_{\geq 3}} \geq 
\frac{\sqrt{N}}{\beta}\left( \sqrt{C'} - \sqrt{\frac{C''}{\beta}}\right)\, ,
$$
which leads immediately to the claimed estimate \eqref{est:sJm}.
\end{proof}

\section{Proof of the Main Results}
In this section we give the proofs of the main theorems of our paper.

\subsection{Proof of Theorem \ref{thm:goal}}
The proof is a straightforward application of Proposition \ref{lem:numeden} and \ref{prop:sJm}.
	{Having fixed $m \in \N$,  we apply \eqref{cheb} with $\Phi = J^{(m)}$ and $\lambda = \delta_1$ to get 
\begin{align}
\label{eq:last_neq_thmgoal}
\bP\Big(\abs{J^{(m)}(t) - J^{(m)}(0)} \geq \delta_1 { \sigma_{J^{(m)}(0)}} \Big)
& \leq C_0  \left(\frac{|\tb - 1|^2}{\beta^2} + \frac{C_1}{\beta^{3}}\right)\frac{ t^2 }{\delta_1^2}
\end{align}
from which one deduces the the statement of Theorem \ref{thm:goal}.}


\subsection{Proof of Theorem \ref{thm:main} and Theorem \ref{thm:main2}}
{ The proofs of Theorem \ref{thm:main} and Theorem \ref{thm:main2} are quite similar and we develop them at the same time.
As in the proof of Theorem \ref{thm:goal},  the first step is to use  Chebyshev inequality to bound 
\begin{equation}
	\label{eq:chebicev0}
	\begin{split}
		\bP\left( \abs{\Phi(t) - \Phi} > \lambda \sigma_{\Phi} \right)
	\leq \frac{1}{\lambda^2} \frac{\sigma^2_{\Phi(t) - \Phi}}{ \sigma^2_{\Phi}} \,,
	\end{split} 
	\end{equation}
	where the time evolution is intended with respect to the FPUT flow $\phi^t_{F}$ or the Toda flow  $\phi^t_{T}$.
	Accordingly, the probability is  calculated  with respect to  the  FPUT Gibbs measure \eqref{eq:misura_vera} or the Toda   Gibbs measure  \eqref{eq:misura_vera_toda}.

Next we observe that the quantity $\Phi := \sum_{j=1}^{N-1} \wh g_j E_j $ defined in \eqref{Phi} can be written in the form
\begin{equation}
\label{Phi_rel}
\Phi(\bp,\br) = \sum_{j=1}^{N-1} \wh g_j E_j = 
\frac{1}{2 \sqrt{N}}\sum_{j,l=0}^{N-1} g_l \left( p_j p_{j+l} + r_j r_{j+l} \right) =\frac{1}{2 \sqrt{N}} G_2(\bp,\br),
\end{equation}
where  $\bg\in\R^N$ is a $m$-admissible vector and $G_2(\bp,\br)$ is a $m$-admissible function of the second kind, as in  Definition~\ref{def:Gad}. As the inequality \eqref{cheb} is scaling invariant, proving \eqref{eq:chebicev0} is equivalent to obtain that
\begin{equation}
	\label{eq:chebichev}
	\begin{split}
		\bP\left( \abs{G_2(t) - G_2} > \lambda \sigma_{G_2} \right)
	\leq \frac{1}{\lambda^2} \frac{\sigma^2_{G_2(t) - G_2}}{ \sigma^2_{G_2}} 
	\end{split} .
	\end{equation}
	 Applying Proposition \ref{prop:ICTP} we can estimate  $\sigma^2_{G_2}$.   We are  then  left to  give an  upper bound to $\sigma^2_{G_2(t) - G_2}$. 
By Lemma \ref{G.lin.comb}, there exists a unique sequence $\{c_j\}_{j=0}^{\wt m -1}$,  with $\max_j |c_j| $ independent from $N$,  such that 
 	$G_2(p,r) = \sum_{l=0}^{\wt m-1} c_lJ_{2}^{(2l+2)}$,  where $ J_{2}^{(2l+2)}$ are defined in \eqref{Jm2.struct}. Hence we bound 
\begin{equation*}
\label{}
	\sigma_{G_2(t) - G_2(0)} \leq \sum_{l=0}^{\wt m-1} |c_l| \,  \sigma_{J_{2}^{(2l+2)}(t) - J_{2}^{(2l+2)}(0)}  .
\end{equation*}

Next we interpolate  $J_{2}^{(2l)}$ with the  integrals $J^{(2l)}$ and exploit the fact that they are adiabatic invariants for the FPUT   flow and integrals of motion for the Toda flow.
More precisely
\begin{align}
\label{var.J.T}
\sigma_{J_{2}^{(2l)}(t) - J_{2}^{(2l)}(0)} & \leq 
 \sigma_{J_{2}^{(2l)}(t) - J^{(2l)}(t)}  +\sigma_{J^{(2l)}(0) - J_{2}^{(2l)}(0)}\\
 \label{var.J.F}
&\quad + \sigma_{J^{(2l)}(t) - J^{(2l)}(0)}.
\end{align}
By the invariance of the two   measures with respect to  their corresponding flow  and Lemma \ref{cor:stima_coda}, we get both for FPUT and Toda the estimate 
\begin{equation}
\label{}
 \sigma_{J_{2}^{(2l)}(t) - J^{(2l)}(t)} =  \sigma_{J_{2}^{(2l)}(0) - J^{(2l)}(0)} =  \sigma_{J_{\geq 3}^{(2l)}}  \leq \sqrt{\frac{\tilde{C}_1 N}{\beta^{3}}}, 
\end{equation}
for some constant $\tilde{C}_1>0$ and for $\beta>\beta_0$ and $N>N_0$.
As \eqref{var.J.F} is zero for the Toda flow  (being $J^{(2l)}(t)$ constant along the flow), we get 
\begin{equation}
	\label{var.G.T}
	\sigma_{G_2\circ \phi_T^t  -  G_2}^2 \leq \frac{C_1 N}{\beta^3},
	\end{equation}
	for some constant $C_1>0$ and for $\beta>\beta_0$ and $N>N_0$.
{ Combing Proposition \ref{prop:ICTP} with \eqref{var.G.T}  we conclude that 
\begin{equation}
	\label{eq:last_ineq_thm2}
	\begin{split}
		\bP\left( \abs{G_2\circ \phi^t_T  - G_2} > \delta_1 \sigma_{G_2} \right)
	\leq \frac{C_1}{\delta_1^2\beta},\quad  \forall \delta_1>0,
	\end{split} 
	\end{equation}
namely we have concluded the proof of Theorem~\ref{thm:main2}.}

We are left  to estimate \eqref{var.J.F} for FPUT, but this is exactly the quantity bounded in Proposition  \ref{lem:numeden}. 
We conclude that 
\begin{equation}
	\label{var.G.F}
	\sigma^2_{G_2\circ \phi_F^t -  G_2} \leq  \frac{C_1 N}{\beta^3} + C_3 N \left(\frac{|\tb - 1|^2}{\beta^4} + \frac{C_2}{\beta^{5}}\right) t^2,
	\end{equation}
	for some constant $C_j>0$, $j=1,2,3$ and for $\beta>\beta_0$ and $N>N_0$.

Combing Proposition \ref{prop:ICTP} with \eqref{var.G.F} we obtain
\begin{equation}
	\label{G2_est}
	\bP\left( \abs{G_2\circ \phi^t_F - G_2} > \lambda \sigma_{G_2} \right)\leq 
	\frac{C_1 }{\lambda^2\beta} + \frac{C_3}{\lambda^2}  \left(\frac{|\tb - 1|^2}{\beta^2} + \frac{C_2}{\beta^{3}}\right) t^2.
		\end{equation}
Choosing $\lambda=\beta^{-\varepsilon}$ with $0<\varepsilon<\frac{1}{4}$,  \eqref{G2_est} is equivalent to 

\begin{equation}
	\label{eq:final_fput}
	\bP\left( \abs{G_2\circ \phi^t_F - G_2} > \frac{\sigma_{G_2}}{\beta^\varepsilon} \right)\leq 
	 \frac{C_1}{\beta^{2\varepsilon}}  ,
	\end{equation}
		  for some redefine constant $C_1>0$ and for every time $t$ fulfilling  \eqref{time2}.


 We have thus concluded the proof of Theorem~\ref{thm:main}.
 
 }

\appendix

\section{Proof of Lemma~\ref{JM.STRUC} }
\label{app:struc}
{
In order to prove Lemma~\ref{JM.STRUC} we  describe more specifically the  Toda integrals and characterize their quadratic parts.
Equation  \eqref{J_cyclic} follows by the explicit expression of $h_j^{(m)}$ in \eqref{eq:general_super_motzkin}, as the coefficients $\rho^{(m)}(\bn, \bk)$ do not depend on the index $j$.
We recall that $h_1^{(m)}$ takes the form
$$
h_1^{(m)}(\bp,\br) = \sum_{( \bk, \bn) \in \cA^{(m)} } (-1)^{|\bk| } \rho^{(m)}(\bn, \bk) \, \bp^{\bk} e^{-\bn ^\intercal \br}  \, ,
$$
with
\begin{equation*}
	  {\rm supp }\, \bk , \ \ {\rm supp }\, \bn  \subseteq B^\td_\wt m(0):= \{ j \colon \td(0,j) \leq \wt m \}  , 
	   \qquad |\bk| + 2|\bn|  = m . 
\end{equation*}
In particular it is clear that $h_1^{(m)}$ has diameter $2\wt m \leq m$.

Now we  Taylor expand around $\br=0$ the exponential with integral remainder:
\begin{equation*}
e^{- \bn ^\intercal \br} = 1 - \bn ^\intercal \br + \frac12 (\bn ^\intercal \br)^2 + \frac{(\bn ^\intercal \br)^3}{2} \int_0^1 e^{-s\bn^\intercal \br} \, (1-s)^2 \, \di s
\end{equation*}
and we substitute it  in $h_1^{(m)}$, obtaining an expansion of the form:
\begin{equation*}
h_1^{(m)}(\bp,\br) = \sum_{ (\bk, \bn) \in \cA^{(m)} } (-1)^{|\bk| } \rho^{(m)}(\bn, \bk) \, \bp^{\bk} \left( 1 - \bn ^\intercal \br + \frac12 (\bn ^\intercal \br)^2 + \frac{(\bn ^\intercal \br)^3}{2} \int_0^1 e^{-s\bn^\intercal \br} \, (1-s)^2 \, \di s \right)\,.
\end{equation*} 
We can rewrite the above expression in the form
\begin{equation*}
h_1^{(m)}(\bp,\br) = \vf_0^{(m)} + \vf_1^{(m)}(\bp,\br) + \vf_2^{(m)}(\bp,\br) + \vf_{\geq 3}^{(m)}(\bp,\br) \, ,
\end{equation*}
where $\vf_\ell^{(m)}$, $\ell = 0,1,2$, are the Taylor polynomials at $(\bp, \br) = (\boldsymbol{0},\boldsymbol{0})$. Their explicit expressions read
\begin{align*}
& \vf_0^{(m)} = \sum_{(\boldsymbol{0}, \bn) \in \cA^{(m)}} \, \rho^{(m)}(\bn, \boldsymbol{0}) \, , \quad 
\vf_1^{(m)} = - \sum_{(\boldsymbol{0}, \bn) \in \cA^{(m)}} \, \rho^{(m)}(\bn,\boldsymbol{0}) \bn^\intercal \br \, - \sum_{(\bk, \bn )\in \cA^{(m)} \atop |k| = 1 } \, \rho^{(m)}(\bn,\bk)\bp^\bk, \\
& \vf_2^{(m)} = \sum_{(\boldsymbol{0}, \bn) \in \cA^{(m)}} \, \rho^{(m)}(\bn,\boldsymbol{0}) \frac{(\bn^\intercal \br)^2}{2} + \sum_{(\bk, \bn) \in \cA^{(m)} \atop |k| = 1 } \, \rho^{(m)}(\bn,\bk)\bp^\bk\bn^\intercal \br + \sum_{(\bk, \bn) \in \cA^{(m)} \atop |k| = 2 } \, \rho^{(m)}(\bn,\bk)\bp^\bk.
\end{align*}

We deduce from these explicit formulas that if $m$ is odd then $\vf_0^{(m)} \equiv 0$ as well as the first sum defining $\vf_1^{(m)}$ and the first and last one defining $\vf_2^{(m)}$.  Indeed the sums  are carried on  an  empty set. If $m$ is even the second sum defining $\vf_1^{(m)}$ and the second one defining $\vf_2^{(m)}$ are zero for the same reason.  
Concerning $\vf_{\geq 3}^{(m)}$, it has a zero of order greater { than } 3 in the variables $(\bp,\br)$, and it has the form
$$
\vf_{\geq 3}^{(m)}(\bp,\br)  := 
\sum_{(\bk, \bn) \in \cA^{(m)} \atop |k| \geq 3 } \, (-1)^{|\bk| } \rho^{(m)}(\bn, \bk) \, \bp^{\bk} \left( 1 - \bn ^\intercal \br + \frac12 (\bn ^\intercal \br)^2 + \frac{(\bn ^\intercal \br)^3}{2} \int_0^1 e^{-s\bn^\intercal r} \, (1-s)^2 \, \di s \right) \, .
$$
These, together with the explicit formula of $\rho^{(m)}(\bn,\bk)$,  prove \eqref{hj2}.

It is easy to see that defining
\begin{align*}
& J_0^{(m)} := \frac{1}{m}\sum_{j = 0}^{N-1} S_j \vf_0^{(m)}, \qquad J_1^{(m)} := \frac{1}{m} \sum_{j = 0}^{N-1} S_j \vf_1^{(m)}, \\
& J_2^{(m)} := \frac{1}{m}\sum_{j = 0}^{N-1} S_j \vf_2^{(m)}, \qquad
 J_{\geq 3}^{(m)} := \frac{1}{m}\sum_{j = 0}^{N-1} S_j \vf_{\geq 3}^{(m)},
\end{align*}

we immediately get that
$$
J^{(m)} = J_0^{(m)} + J_1^{(m)} + J_2^{(m)} + J_{\geq 3}^{(m)}.
$$

Clearly  $J_0^{(m)}$ it is a constant that is zero for $m$ odd; moreover thanks to the boundary condition \eqref{media} and the linearity of $J_1^{(m)}$ we have that $J_1^{(m)} = 0$.
Further,  $J_{\geq 3}^{(m)}$ is clearly a cyclic function.
In order to get \eqref{Jm2.struct} and \eqref{ab}  for $J_2^{(m)}$ we have to split the proof in two different cases.

\paragraph{Case $m$ odd.}

In this case thanks to the property of $\vf_2^{(m)}$, the definition of $J_2^{(m)}$ and \eqref{cyc.quad2} we  get that there exists a cyclic and symmetric matrix $B^{(m)}$ such that:
$$
J_2^{(m)} = \bp^\intercal B^{(m)} \br  .
$$
Moreover since the ${\rm diam} (\bk), {\rm diam }(\bn)$ defining $\vf^{(m)}_2$ are at most $\wt m$ (see Remark \ref{rem:diam.hj}) we have that the vector $\tb^{(m)}$ representing the matrix $B^{(m)}$ is $m$-admissible and from \eqref{rhom} we have that $\tb^{(m)}_j = \tb^{(m)}_{N-j}$ are positive integers for all  $j=0,\ldots, \wt m$.
\paragraph{Case $m$ even.} 
As before  there exist two matrices $A^{(m)}, D^{(m)}$ represented by $m$-admissible vectors such that:
\begin{equation*}
J_2^{(m)} = \bp^\intercal A^{(m)} \bp + \br^\intercal D^{(m)} \br \, , \, \quad
\ta_k^{(m)} = \ta_{N-k}^{(m)} \in \N \, ,
\quad \td_k^{(m)} = \td_{N-k}^{(m)} \in \N  \, ,
 \quad 0 \leq k \leq \wt m . 
\end{equation*}
 
We have just to prove that the two matrices are equal; to do this we {exploit the involution property} of the  Toda integrals.
Indeed we know that   $\left\{ J^{(j)}, J^{(k)}\right\} = 0, $ for any $j,k$.
It follows easily that  also their quadratic parts must commute:
\begin{equation}
\label{eq:com.quad}
 \left\{ J^{(k)}_{2}, J^{(j)}_{2}\right\} = 0, \qquad \forall \, k,j .
\end{equation}
To compute explicitly the  Poisson bracket we change coordinates via the  Hartley transform  \eqref{dht} getting that: 
	\begin{align*}
	J^{(m)}_2 &= \sqrt{N}\sum_{j=1}^{N-1} \wh a_j\wh p_j^2 + \wh d_j\wh r_j^2 = \sqrt{N}\sum_{j=1}^{N-1} \wh a_j\wh p_j^2 + \wh d_j\omega_j^2\wh q_j^2, \\
	J_2^{(2)} &=\frac{1}{2} \sum_j\wh p_j^2 + \omega_j^2 \wh q_j^2 , 
	\end{align*}
	where $\omega_j = 2\sin\left(\pi\frac{ j}{N}\right)$.
As the Hartley transform is a symplectic map, 	by \eqref{eq:com.quad} we get 
\begin{equation}
0 = \left\{ J^{(2)}_{2}, J^{(m)}_{2}\right\} = \sqrt{N}\sum_{j=1}^{N-1} \omega_j^2\left(\wh a_j - \wh d_j\right)\wh p_j \wh q_j ,
\end{equation}  
which implies that $\wh a_j=\wh d_j$ for all $j\ne0$. To prove that also $\wh a_0 = \wh d_0$ we come back to the original variables getting that:

\begin{equation}
\begin{split}
		 &\ta^{(m)}_j =\frac{1}{\sqrt{N}} \wh a_0 + \frac{1}{\sqrt{N}} \sum_{k=1}^{N-1} \wh a_j \left(\cos\left(2\pi \frac{jk}{N}\right) + \sin\left(2\pi \frac{jk}{N}\right) \right), \\ &\td^{(m)}_j = \frac{1}{\sqrt{N}} \wh d_0 + \frac{1}{\sqrt{N}} \sum_{k=1}^{N-1} \wh a_j \left(\cos\left(2\pi \frac{jk}{N}\right) + \sin\left(2\pi \frac{jk}{N}\right) \right) ,
\end{split}
\quad \forall \, j.
\end{equation}

This means that $ \ta^{(m)}_j -  \td^{(m)}_j = \frac{\wh a_0 - \wh d_0}{\sqrt{N}}$ for all $j=0,\dots,N-1$.
Since  $\ta^{(m)},\td^{(m)}$ are $m$-admissible it follows that $\ta^{(m)}_{\wt m + 1}=\td^{(m)}_{\wt m + 1}=0$ so that 
 $$\frac{\wh a_0 - \wh d_0}{\sqrt{N}} = \ta^{(m)}_{\wt m + 1} -  \td^{(m)}_{\wt m + 1} = 0,$$
which proves the statement.
}
\section{Proof of Lemma \ref{LEM:ORDINE_BETA}}
\label{app:process}
We prove the lemma  for both  the   FPUT and Toda measure.

First of all we observe  that for $d,v=2,3$:

\begin{equation}
\frac{1}{4^d}\prod_{j\in {\rm Supp} \, \bn}\min\left(e^{-dn_j r_j},1 \right)\leq \left(\int_{0}^1  e^{- \xi \bn^\intercal\br  }(1-\xi)^v\di \xi\right)^d \leq \frac{1}{3^d}\prod_{j\in {\rm Supp} \, \bn}\max\left(e^{-dn_j r_j},1 \right) \, .
\end{equation}

This means that we have actually to prove that for  any fixed  multi-index $\bk,\bl,\bn\in \N_0^{N}$ there exist two constants $C_{\bk,\bl}^{(1)} \in \R$ and $C_{\bk,\bl}^{(2)}>0$ such that:

\begin{align}
& \la \bp^\bk\br^\bl\prod_{j\in {\rm Supp} \, \bn}\min\left(e^{-n_j r_j},1 \right) \ra_\theta \geq C^{(1)}_{\bk,\bl}\beta^{-\frac{|\bk| + |\bl | }{2}} \, , \\
& \la \bp^\bk\br^\bl\prod_{j\in {\rm Supp} \, \bn}\max\left(e^{-n_j r_j},1 \right) \ra_\theta \leq  C^{(2)}_{\bk,\bl}\beta^{-\frac{|\bk| + |\bl | }{2}}\, .
\end{align}

Moreover since for the two measures $\di \mu_{F,\theta}, \di \mu_{T,\theta}$ all $\bp$ and $\br$ are independent random variables and moreover the $p_j$ are independent and normally  distributed according to $\cN(0,\beta^{-1})$, it follows
\begin{align}
& \la \bp^\bk\br^\bl\prod_{j\in {\rm Supp} \, \bn}\min\left(e^{-n_j r_j},1 \right) \ra_\theta =
 \la \bp^\bk \ra_\theta\la\br^l\prod_{j\in {\rm Supp} \, \bn}\min\left(e^{-n_j r_j},1 \right) \ra_\theta\\
& \la \bp^\bk\br^\bl\prod_{j\in {\rm Supp} \, \bn}\max\left(e^{-n_j r_j},1 \right) \ra_\theta =\la \bp^\bk \ra_\theta \la\br^\bl\prod_{j\in {\rm Supp} \, \bn}\max\left(e^{-n_j r_j},1 \right) \ra_\theta 
\end{align}
where 
\begin{equation}
\la \bp^\bk \ra_\theta =\la\prod_{i}p_i^{k_i} \ra_\theta =\begin{cases} \displaystyle{\prod_{i}\frac{(k_i-1)!!}{\beta^{\frac{k_i}{2}}} },&\quad k_i \text{ all even} \\ 0, &\quad  \mbox{ some $k_i$  odd}
\end{cases}
\end{equation}
Here $k!!$ denotes the double factorial.
Instead the distribution of the $r_j$ is different for the two measures, so we need to  calculate it separately for the FPUT and Toda chain.
\paragraph{FPUT chain.}

Let's start considering $ \la r^{l}\min\left(e^{-n r},1 \right) \ra_\theta$:
\begin{equation}
\begin{split}
\la r^{l}\min\left(e^{-n r},1 \right) \ra_\theta &=
\frac{\int_{\R^-}r^l e^{-\theta r -\beta \left(\frac{r^2}{2} + \frac{r^3}{3} + \frac{r^4}{4}\right)} \di r + \int_{\R^+} r^le^{-nr}e^{-\theta r -\beta \left(\frac{r^2}{2} + \frac{r^3}{3} + \frac{r^4}{4}\right)} \di r}{\int_{\R} e^{-\theta r -\beta \left(\frac{r^2}{2} + \frac{r^3}{3} + \frac{r^4}{4}\right)} \di r} \\
&= \beta^{-\frac{l}{2}} \frac{\int_{\R^-}r^l e^{-\frac{\theta}{\sqrt{\beta}} r - \left(\frac{r^2}{2} + \frac{r^3}{3\sqrt{\beta}} + \frac{r^4}{4\beta}\right)} \di r + \int_{\R^+} r^le^{-\frac{n}{\sqrt{\beta}}r}e^{-\frac{\theta}{\sqrt{\beta}} r -\left(\frac{r^2}{2} + \frac{r^3}{3\sqrt{\beta}} + \frac{r^4}{4\beta}\right)} \di r}{\int_{\R} e^{-\frac{\theta}{\sqrt{\beta}} r - \left(\frac{r^2}{2} + \frac{r^3}{3\sqrt{\beta}} + \frac{r^4}{4\beta}\right)} \di r} \\
& \geq \beta^{-\frac{l}{2}} \frac{\int_{\R^-}r^l e^{-\frac{\theta}{\sqrt{\beta}} r - \left(\frac{r^2}{2} + \frac{r^3}{3\sqrt{\beta}} + \frac{r^4}{4\beta}\right)} \di r}{\int_{\R} e^{-\frac{\theta}{\sqrt{\beta}} r - \left(\frac{r^2}{2} + \frac{r^3}{3\sqrt{\beta}} + \frac{r^4}{4\beta}\right)} \di r} \, .
\end{split}
\end{equation}

Since for $\beta$ large enough $\theta(\beta)$ is uniformly bounded, it follows that  there is a positive constant $C_{l}$ such that:
\begin{equation}
\label{eq:lower_beta}
\la r^{l}\min\left(e^{-n r},1 \right) \ra_\theta \geq (-1)^l \frac{C_{l}}{\beta^\frac{l}{2}} \, .
\end{equation}

We notice that if $l$ is even then the right end side of \eqref{eq:lower_beta} is positive. 
The proof for  $\la r^{l}\max\left(e^{-n r},1 \right) \ra_\theta $ follows in the same way so we get the claim for the FPUT chain.

\qed 
\paragraph{Toda chain.}

For the Toda chain the computation is a little bit more involved, so we prefer to split it in different parts.

\begin{lemma}
	\label{lem:mom_r}
	Consider the measure \ref{eq:misura_falsa}, then there exists a $\beta_0> 0$ such that for all $\beta>\beta_0$ there exists $\theta \equiv \theta(\beta)\in [1/3, 2]$ such that  
	\begin{equation}
	\la r_j^k \ra_\theta  = \begin{cases}
	0 &\quad k=1\\
	\cO\left(\frac{1}{\beta^{\frac{k}{2}}}\right) &\quad k \ne 1 
	\end{cases} \, .
	\end{equation}
\end{lemma}

\begin{proof}
	First we prove that, for  any $\beta$ large enough, we can chose $ \theta(\beta)$ in a compact interval $\cI$ such that  $\la r_j \ra_\theta = 0$. We notice that: 
	
	\begin{equation}
	\label{eq:general_rmom_t}
	\la r^k \ra_\theta = (-1)^k \frac{\d_\theta^k \int_\R e^{-(\theta + \beta) r - \beta e^{-r}}\di r}{\int_\R e^{-(\theta + \beta) r - \beta e^{-r}}\di r} \overset{\left(e^{-r} = x\right)}{=} (-1)^k\frac{\d_\theta^k \int_{\R^+} x^{\theta + \beta - 1}e^{-\beta x} \di x}{\int_{\R^+} x^{\theta + \beta - 1}e^{-\beta x} \di x} = (-1)^k\frac{\d_\theta^k \frac{\Gamma(\beta + \theta)}{\beta^\theta}}{\frac{\Gamma(\beta + \theta)}{\beta^\theta}} \, ,
	\end{equation}	
	where $\Gamma(z)$ is the usual Gamma function and we used the following equality:
	$$\int_0^\infty 
	t^{z-1} e^{-xt} \di t = \frac{\Gamma(z)}{x^z}\, . $$	
	In the case $k=1$ one obtains
	\begin{equation}
	\label{eq:meannzero}
	\la r \ra_\theta   =\log(\beta)- \frac{\Gamma'(\theta + \beta)}{\Gamma(\theta + \beta)}\,.
		\end{equation}
	{ 	Introducing the digamma function  $\psi(z)=\frac{\Gamma'(z)}{\Gamma(z)}$ \cite{LukeSpecial} and 
	 using the inequality
	$$ \log x-{\frac {1}{x}} \leq  \psi(x)\leq  \log x-{\frac {1}{2x}} , \qquad \forall x >0 , $$
	it is easy to show that there exists $\beta_0 >0$ such that $\forall \beta > \beta_0$ one has 
	$$
	\psi\left(\frac{1}{3} + \beta\right) \leq \log\left(\frac{1}{3} + \beta\right) - \frac{1}{2(1/3 + \beta)} \leq \log \beta 
	$$
	and
	$$
	\psi(2+ \beta) \geq \log(2+\beta) - \frac{1}{2+\beta} \geq \log \beta . 
	$$
	Since  $x \mapsto \psi(x)$ is continuous on $(1, + \infty)$, by the intermediate value theorem there exists  $\theta(\beta) \in [1/3, 2]$ fulfilling
$\psi(\theta + \beta)=\log\beta$	  which implies by \eqref{eq:meannzero} that 
\begin{equation}
\label{psi=log} 
 \la r_j\ra_{\theta} =\log(\beta)- \frac{\Gamma'(\theta + \beta)}{\Gamma(\theta + \beta)}= 0.
 \end{equation} 
 }
	We will prove the remaining part  of the claim by induction; \eqref{eq:general_rmom_t}	leads in the case $k=2$ to:
	\begin{equation}
	\begin{split}
	\la r^2 \ra_\theta& = \frac{\beta^\theta}{\Gamma(\theta +\beta)}\d_\theta\left(\frac{\Gamma'(\theta + \beta) - \ln(\beta)\Gamma(\theta + \beta)}{\beta^\theta}\right) \\
	& = \frac{\beta^\theta}{\Gamma(\theta +\beta)}\d_\theta\left(\frac{\beta^\theta}{\Gamma(\theta +\beta)}\left(\psi(\theta +\beta) - \ln(\beta)\right) \right) \\
	& = \la r_j\ra_\theta\left(\psi(\theta +\beta) - \ln(\beta)\right) + \psi^{(1)}(\theta + \beta) \\ &= \psi^{(1)}(\theta + \beta),
	\end{split}
	\end{equation}
	where $\psi^{(s)}$ is the $s^{th}$ polygamma function defined as 	
	$\psi^{(s)}(z) := \dfrac{\d^s \psi(z) }{\d z ^s}$. 
	For $x\in \R$  it has the following expansion as  $x\to +\infty$ :
	\begin{equation}
	\label{eq:psiexpansion}
	\psi^{(s)}(x) \sim (-1)^{s+1}\sum_{k=0}^{\infty}\frac{(k+s-1)!}{k!}\frac{B_k}{x^{k+s}}, \qquad s \geq 1\, ,
	\end{equation}
	where $B_k$ are the Bernoulli number of the second kind.
	Therefore
	\[
	\la r^2 \ra_\theta=\psi^{(1)}(\theta + \beta) \overset{\beta > \beta_0}{=} \cO\left(\frac{1}{\beta}\right).
	\]
	
	So the first inductive step is proved.
	Next suppose the statement true for $k$ and let us prove it for $k+1$.
		
	\begin{equation}
	\begin{split}
	\la r^{k+1} \ra_\theta& = (-1)^{k+1}\frac{\beta^\theta}{\Gamma(\theta +\beta)}\d^k_\theta\left(\frac{\Gamma'(\theta + \beta) - \ln(\beta)\Gamma(\theta + \beta)}{\beta^\theta}\right) \\
	& = (-1)^{k+1}\frac{\beta^\theta}{\Gamma(\theta +\beta)}\d^k_\theta\left(\frac{\beta^\theta}{\Gamma(\theta +\beta)}\left(\psi(\theta +\beta) - \ln(\beta)\right) \right) \\
	& = (-1)^{k+1}\frac{\beta^\theta}{\Gamma(\theta +\beta)}\d^k_\theta\left(\frac{\beta^\theta}{\Gamma(\theta +\beta)}\right)\left(\psi(\theta +\beta) - \ln(\beta)\right) \\ &
	+(-1)^{k+1}\frac{\beta^\theta}{\Gamma(\theta +\beta)}\sum_{n=1}^{k}\binom{k}{n} \d^{k-n}_\theta\left(\frac{\beta^\theta}{\Gamma(\theta +\beta)}\right)\d_\theta^n\psi(\theta +\beta) \\
	& = 0 + \sum_{n=1}^{k}\binom{k}{n}(-1)^{n+1} \la r^{k-n}\ra_\theta \d_\theta^n\psi(\theta +\beta) = \cO\left( \frac{1}{\beta^{\frac{k}{2}}} \right)\, ,
	\end{split}
	\end{equation}	
	where we used \eqref{psi=log} and \eqref{eq:psiexpansion}.
\end{proof}
We are now ready to prove the last part of Lemma \ref{LEM:ORDINE_BETA} for the Toda chain:

\begin{equation}
\begin{split}
\la r^l\max(1,e^{-nr}) \ra_\theta &= \frac{\int_{\R^+} r^l e^{-(\theta + \beta) r -\beta e^{-r}} \di r + \int_{\R^-} r^l e^{-(\theta+ \beta -n) r -\beta e^{-r}} \di r}{ \int_{\R} e^{-\theta r -\beta e^{-r}} \di r} \\ 
& \leq  \frac{\int_{\R^+} r^l e^{-(\theta + \beta) r -\beta e^{-r}} \di r}{ \int_{\R} e^{-(\theta + \beta) r -\beta e^{-r}} \di r}\, .
\end{split}
\end{equation}

The last integral can be estimated in the same way  as in the previous lemma, moreover the lower bound follows in the same way, so we get the claim also for the Toda chain.
\qed

\section{Measure Approximation}
\label{app:meas.app}
In this section we show how to approximate the measure $\di \mu$, in which the variables are constrained, 
with the measure $\di \mu_\theta$, where  all variables are independent. 
The proof  follows the construction of \cite{BCM14} {(where it is done for Dirichlet boundary conditions)} which  applies   both to  the Gibbs measure of  FPUT \eqref{eq:misura_vera} and  Toda \eqref{eq:misura_vera_toda}.
To simplify the construction  we consider a general potential  $V\colon \R\to\R$ and make the following assumptions:
\begin{itemize}
\item[ (V1)] There exist $\beta_0 >0$ and a compact interval $\cI \subset \R$ such that for any $\beta > \beta_0$, there exists $\theta\equiv \theta(\beta) \in \cI$ such that
\begin{equation}
\label{V1}
\int_{\R} r \, e^{- \theta r  - \beta V(r)} \, \di r  = 0 . 
\end{equation}
\item[(V2)] There exist $\beta_0, \tC_1 , \tC_2 >0$ such that for any $\beta > \beta_0$, with $\theta = \theta(\beta)$ of (V1), one has 
\begin{equation}
\label{V2}
\frac{\tC_1}{\beta^{k/2}} < \int_{\R} |r|^k \, e^{- \theta r  - \beta V(r)} \, \di r < \frac{\tC_2}{\beta^{k/2}} , \qquad k = 0, \ldots, 4 .
\end{equation}
In particular the moments up to  order $4$ are finite.
\item[(V3)] There exists $\beta_0 >0$ such that $\forall \beta > \beta_0$, with  $\theta = \theta(\beta)$ of (V1), one has 
\begin{equation}
\label{V3}
\inf_{r \in \R} \abs{\theta r + \beta V(r)} > - \infty ,
\end{equation}
namely the function $r \mapsto \theta r + \beta V(r)$ is bounded from below.
\end{itemize}

Both   the FPUT potential $V_F(x)$ and the Toda potential  $V_T(x)$ satisfy the assumptions (V1)--(V3) by the results  of Appendix \ref{app:process}.
 
We define the constraint measure $\di \mu^V$ on the restricted phase space $\cM$ as
\begin{equation}\label{eq:misura_vera.ab}
\di \mu^V :=    \frac{1}{Z_V(\beta)} \  e^{-\beta \sum_{j=1}^N \frac{p_j^2}{2}} \ 
e^{-\beta \sum_{j=1}^N V(r_j)}   \ \delta\left(\sum_j r_j =0\right) \ \delta\left(\sum_j p_j =0\right)  \ \di \bp \,  \di \br, 
\end{equation}
and the unconstrained measure $\di \mu^V_{\theta}$  on the extended phase space $\R^N \times \R^N$ as
  \begin{equation}\label{eq:misura_falsa.ab}
\di \mu^V_{\theta} :=    \frac{1}{Z_{V, \theta}(\beta)} \  e^{-\beta \sum_{j=1}^N p_j^2/2} \ 
e^{-\beta \sum_{j=1}^N V(r_j)}  \, e^{- \theta \sum_{j=1}^N r_j }  \ \di \bp \,  \di \br ;
\end{equation}
  as usual $Z_V(\beta)$ and $Z_{V, \theta}(\beta)$ are the   normalizing constants  of $\di \mu^V, \,\di \mu^V_\theta$ respectively . 
We denote the expectation of $f$ with respect to the measure
$\di \mu^V$ as  $\la f \ra_V$, and with respect to the measure $\di \mu^V_\theta$ as $\la f \ra_{V, \theta}$. \\
We also  denote by  $\displaystyle{\norm{f}_{V, \theta}:= \la f^2 \ra_{V, \theta}^{1/2}}$ the $L^2$ norm of $f$ with respect to the measure $\di \mu_\theta^V$.

The main result is the following one:
\begin{theorem}
\label{magic.abs}
Assume that (V1)--(V3) hold true. Fix $K \in \N$ and assume that  $f \colon \R^N\times \R^N \to \R$ have support of size $K$ (according to definition \ref{def:supp}) and finite second order moment with respect to $\di\mu^V_\theta$.
Then there exist $C, N_0$ and $\beta_0 $ such that
for all $N > N_0$, $\beta > \beta_0$ one has  
\begin{equation}
\label{eq:magic:ab}
\abs{\la f \ra_V - \la f \ra_{V,\theta} }\leq  C \frac{K}{N} \sqrt{\la f^2\ra_{V,\theta} - \la f \ra_{V,\theta}^2} . 
\end{equation}
\end{theorem}

  \subsection{Proof of Theorem \ref{magic.abs}}
  \label{proof_magic_abs}
Introduce the structure function
\begin{equation}
\label{struc.func}
\Omega_{N}(x) := \int\limits_{ x_1 + \ldots + x_{N} = x} \, e^{-\beta \sum_{j=1}^{N} V(x_j)} \   \di  x_1 \ldots \di x_{N} , 
\qquad
\forall x \in \R . 
\end{equation}
The important remark is that $\Omega_N(x)$ is $N$-times the convolution of the function $e^{-\beta V(x)}$ with itself thus it is  the density function 
of the sum of $N$ iid random variables distributed as $e^{-\beta V(x)}$.

Next, for $\theta \in \R$, we  define the conjugate distribution
\begin{equation}
\label{UN}
U^{(\theta)}_N(x):= \frac{1}{\left(z_{\theta}(\beta)\right)^N} \  e^{-\theta x} \ \Omega_N( x) , \qquad
z_{\theta}(\beta):= 
\int_{\R}  e^{-\beta V(x) - \theta x}\,  \di x  .
\end{equation}
As before, we remark that $U^{(\theta)}_N(x)$ it is $N$-times the convolution of the function $e^{-\beta V(x) - \theta x}$ with itself thus it is  the density function  of the sum of $N$ iid random variables $\{Y_n^{(\theta)}(\beta)\}_{1 \leq n \leq N}$ distributed as
\begin{equation}
\label{Yntb}
Y_n^{(\theta)}(\beta) \sim Y^{(\theta)} := \frac{1}{z_{\theta}(\beta)}   e^{-\beta V(x)- \theta x}\, \di x ,
\end{equation}
moreover thanks to \eqref{V1} we know that $\la Y^{(\theta)} \ra = 0$.
%

The  central limit theorem says that the rescaled random variable $\displaystyle{\frac{1}{\sigma \sqrt N} \sum_{n=1}^N Y_n^{(\theta)}(\beta)}$ converges in distribution to a normal $\cN(0,1)$. 
We want to apply a more refined version  of this result,  called {\em local central limit theorem}, which describes the asymptotic of this convergence.    

In particular we will use a local central theorem whose  proof can be found in \cite[Theorem VII.15]{petrov}; to state it, we first define the functions
\begin{equation}
\label{def:qu}
\tq_\nu(x) := \frac{1}{\sqrt{2\pi}} e^{- \frac{x^2}{2}} \sum_{\cB(\nu)} \tH_{j+2s}(x) \prod_{d=1}^\nu \frac{1}{k_d !} \left( \frac{\gamma_{d+2}}{(d+2)! \, \sigma^{d+2}} \right)^{k_d}
\end{equation}
where $\tH_j$ is the $j$-th Hermite polynomial, $\gamma_d$ is the $d$-th cumulant\footnote{We recall that $\gamma_d = \sum_{\cC(d)} d!(-1)^{m_1+\ldots+m_d-1}\left(m_1+\ldots+m_d-1\right)!\prod_{l=1}^{d}\frac{\alpha_l^{m_l}}{m_l!(l!)^{m_l}}$ where $\alpha_l$ is the $l^{th}$ moment of the random variable and $\cC(d)$ is the set of all non-negative integer solution of $\sum_l lm_l = d$. } of $ Y_n^{(\theta)}(\beta)$, and ${\cB(\nu)}$ is the set of all  non-negative
integer solutions $k_1, \ldots, k_\nu$ of the equalities $k_1 + 2k_2 + \cdots  + \nu k_\nu = \nu$, and
$s = k_1 + k_2 + \cdots +k_\nu$. 
\begin{theorem}[Local central limit]
\label{thm:loc}
Let $\{X_n\}$ be a sequence of iid variables such that 
\begin{itemize}
\item[(i)] For any $1 \leq n \leq N$, one has $\meanval{X_n} = 0$.
\item[(ii)] There exists $k \geq 3$ such that $\meanval{|X_n|^k} < +\infty$ for all $n$. Moreover $\sigma^2:= \meanval{X_n^2}>0$.
\item[(iii)] The random variable $\frac{1}{\sigma \sqrt{N}} \sum_{n=1}^N X_n $ has  a bounded density $\tp_N(x)$.
\end{itemize}
Then there exists $C >0$ such that 
\begin{equation}
\label{thm:loc:est}
\sup_x \abs{\tp_N(x) - \frac{1}{\sqrt{2\pi}} e^{-\frac{x^2}{2}} + \sum_{\nu=1}^{k-2}\frac{\tq_\nu(x)}{N^{\nu/2}}} \leq \frac{C}{N^{(k-2)/2}} , 
\end{equation}
where the $\tq_\nu$'s are defined in \eqref{def:qu}.
\end{theorem}
Applying this theorem in case  $X_n = Y_n^{(\theta)}(\beta)$ , one gets the following result:
\begin{corollary}
Assume (V1)--(V3).
There exist $N_0, \beta_0, C >0$ such that for all $N \geq N_0$, $\beta > \beta_0$ one has
\begin{equation}
\label{UN.exp}
\abs{U^{(\theta)}_N( x) - \frac{1}{\sqrt{2\pi \sigma^2 N}} \exp\left(-\frac{x^2}{2\sigma^2 N}\right) + \sum_{\nu=1}^{2}\frac{\tq_\nu(x/\sigma \sqrt{N})}{N^{(\nu+1)/2}\sigma }} \leq \frac{C}{N^{3/2}\sigma} . 
\end{equation}
\end{corollary}
\begin{proof}
We verify that the assumptions of Theorem \ref{thm:loc} are met in case $X_n = Y_n^{(\theta)}(\beta)$. \\
Item $(i)$ and $(ii)$ hold true thanks to assumptions  (V1) and (V2), in particular $(ii)$ is true with $k=4$. To verify $(iii)$, we  note that  $\displaystyle{\frac{1}{\sigma \sqrt{N}}\sum_{n=1}^N Y_n^{(\theta)}(\beta)}$ has  density given by  $\sigma \sqrt{N} \, U^{(\theta)}_N(\sigma \sqrt{N} x)$. 
This last function is  $N$-times the convolution of  $g_\theta(r):= e^{-\theta r - \beta V(r)}$.
 By assumption  (V3), $g_\theta \in L^\infty(\R)$ and by (V2) it belongs also to $L^1(\R)$. 
 So  Young's convolution inequality implies that   $\sigma \sqrt{N} \, U^{(\theta)}_N(\sigma \sqrt{N} x)$ is bounded uniformly in $x$, hence $(iii)$ of Theorem \ref{thm:loc} is verified.\\
We apply Theorem \ref{thm:loc} with $\tp_N(x) = \sigma \sqrt{N} \, U^{(\theta)}_N(\sigma \sqrt{N} x)$, then rescale the variable $x$ to get \eqref{UN.exp}.
\end{proof}

 We study also the structure function 
\begin{equation*}
\wt \Omega_N(\xi) := \int\limits_{\xi_1 + \ldots + \xi_N = \xi} e^{- \frac{\beta}{2} \sum_{j=1}^N \xi_j^2}\  \di \xi_1 \ldots \di \xi_N . 
\end{equation*}
and the normalized distribution
\begin{equation}
\label{wtUN}
\wt U_N(\xi) :=  \frac{1}{\left(\wt z_{\theta}(\beta)\right)^N} \  \ \wt\Omega_N( \xi) , \qquad
\wt z_{\theta}(\beta):= 
\int_{\R}  e^{-\frac{\beta}{2} \xi^2 }\,  \di \xi  . 
\end{equation}
We have the following result:
\begin{lemma} 
For any $ N \geq 1$, any $\beta >0$, one has
\begin{equation}
\label{wtUNexp}
\wt U_N(\xi) = \sqrt{\frac{\beta}{{2\pi  N}}}\,  \exp\left(-\frac{\beta \xi^2}{2 N}\right) . 
\end{equation}
\end{lemma}
\begin{proof}
The function $\wt U_N$ is the $N$-times convolution of Gaussian functions of the form  $g(\xi) := \sqrt{\frac{\beta}{2\pi}}\, e^{- \frac{\beta}{2} \xi^2}$. 
Since convolution of Gaussians is a Gaussian whose variance is the sum of the variances,  \eqref{wtUN} follows.
\end{proof}
We can finally prove  Theorem \ref{magic.abs}:
\begin{proof}[Proof of Theorem \ref{magic.abs}]
The proof follows closely \cite{BCM14}.
We assume that $f$ is supported on $1, \ldots, K$, the other cases being analogous.
Using that 
$$
Z_{V}(\beta) 
 = \Omega_N(0) \, \wt\Omega_N(0) ,
$$
and  denoting $\wt \bp := (p_1, \ldots, p_K)$ and $\wt \br := (r_1, \ldots, r_K)$, we write 
\begin{align*}
\la f(\wt \bp, \wt \br) \ra_V & = 
\int\limits_{\R^K\times \R^K}
f(\wt \bp, \wt \br) \,  \  
\frac{\Omega_{N-K}\left(-\sum_{j=1}^K r_j\right)}{{\Omega_{N}(0)}} \, 
\frac{\wt \Omega_{N-K} \left(-\sum_{j=1}^K p_k\right)}{\wt \Omega_{N}(0)} \, 
\di \wt \mu
\end{align*}
where
$\di \wt \mu :=  \exp\left({-\beta \sum_{j=1}^K \frac{p_j^2}{2} - \beta\sum_{j=1}^K V(r_j)}\right)\di \wt p \di \wt r$. 
As, by  \eqref{UN} and \eqref{wtUN}, 
\begin{align*}
\frac{\Omega_{N-K}(x)}{{\Omega_{N}(0)}} = \frac{U_{N-K}^{(\theta)}(x)}{U_N^{(\theta)}(0)} 
\frac{e^{\theta x}}{(z_\theta(\beta))^K} , 
\qquad
\frac{\wt \Omega_{N-K} (\xi)}{\wt \Omega_{N}(0)} = 
\frac{\wt U_{N-K}^{(\theta)}(\xi)}{\wt U_N^{(\theta)}(0)} 
\frac{1}{(\wt z_\theta(\beta))^K} , 
\end{align*}
we write  the difference
$\la f \ra_V - \la f \ra_{V,\theta}$ as
$$
\la f \ra_V - \la f \ra_{V,\theta} = \int\limits_{\R^K\times \R^K}
f(\wt \bp, \wt \br) \,  \frac{e^{- \theta \sum_{j=1}^K r_j}}{(z_\theta(\beta))^K \, (\wt z_\theta(\beta))^K}
  \
\tU^{(\theta)}(\wt \bp, \wt \br)
 \, 
\di \wt \mu
$$
where
$$
\tU^{(\theta)}(\wt \bp, \wt \br)
:= 
\frac{U_{N-K}^{(\theta)}\left(-\sum_{j=1}^K r_j\right)}{U_N^{(\theta)}(0)} 
\frac{\wt U_{N-K}^{(\theta)}\left(-\sum_{j=1}^K p_j\right)}{\wt U_N^{(\theta)}(0)} - 1 .
$$
Now we use that
\begin{align*}
\int_{\R^K \times \R^K}  \frac{e^{- \theta \sum_{j=1}^K r_j}}{(z_\theta(\beta))^K \, (\wt z_\theta(\beta))^K}
  \
\tU^{(\theta)}(\wt \bp, \wt \br)
 \, 
\di \wt \mu 
=
\la 1 \ra_V  - \la 1 \ra_{V,\theta} = 0
\end{align*}
so that we can write the difference $\la f \ra_V - \la f \ra_{V,\theta}$   as
$$
\la f \ra_V - \la f \ra_{V,\theta} = \int\limits_{\R^K\times \R^K}
\left( f(\wt \bp, \wt \br) - \la f \ra_{V,\theta} \right) \,  \frac{e^{- \theta \sum_{j=1}^K r_j}}{(z_\theta(\beta))^K \, (\wt z_\theta(\beta))^K}
   \
\tU^{(\theta)}(\wt \bp, \wt \br)
 \, 
\di \wt \mu
$$
Using Cauchy-Schwartz we obtain that
$$
\abs{ \la f \ra_V - \la f \ra_{V,\theta} } \leq \norm{f - \la f \ra_{V, \theta}}_{V, \theta} \ 
\norm{  \tU^{(\theta)}}_{V, \theta} , 
$$
so in order to prove \eqref{eq:magic:ab} we are left to show that uniformly in $N$ and $\beta$ one has 
\begin{equation}
\label{UN.est}
\norm{  \tU^{(\theta)}}_{V, \theta}  \leq C \frac{K}{N} .
\end{equation}
Using \eqref{UN.exp} and \eqref{wtUNexp}, we have that
\begin{align*}
\abs{\frac{U_{N-K}^{(\theta)}\left(x\right)}{U_N^{(\theta)}(0)} 
\frac{\wt U_{N-K}^{(\theta)}\left(\xi\right)}{\wt U_N^{(\theta)}(0)} - 1 } 
\leq C  \left(
\abs{ e^{-\frac{x^2}{2\sigma^2 (N-K)} - \frac{\beta \xi^2}{2(N-K)}} -1 } + 
\frac{N}{(N-K)^{3/2}} \tq_1\left(\frac{x}{\sigma \sqrt{N-K}} \right)
+\frac{K}{N-K} \right)
\end{align*}
Next we use that $|e^{-a^2 - b^2} - 1| \leq a^2 + b^2$,  the explicit expression 
$$
\tq_1(x) = \frac{1}{\sqrt{2\pi}} e^{- \frac{x^2}{2}} (x^3 - 3x)  \frac{\gamma_{3}}{6 \, \sigma^{3}} , 
$$
the estimate 
$\displaystyle{\frac{\gamma_{3}}{6 \, \sigma^{3}}} \leq \tC$ for some $\tC $ independent of $\beta$ (which follows by \eqref{V2} as in our case $\gamma_3 \leq C \beta^{-3/2}$), 
to obtain that there exists $C >0$ such that $\forall N \geq N_0$, $\forall \beta \geq \beta_0$, 
$$
\abs{\frac{U_{N-K}^{(\theta)}\left(x\right)}{U_N^{(\theta)}(0)} 
\frac{\wt U_{N-K}^{(\theta)}\left(\xi\right)}{\wt U_N^{(\theta)}(0)} - 1 } 
\leq  \frac{C}{N} 
 \left(
K+ \beta \xi^2  + 
 \frac{x}{\sigma } + \frac{x^2}{\sigma^2 }+ \frac{x^3}{\sigma^3 N} \right) .
$$
Substituting $x\equiv  -\sum_{j=1}^K r_j$, $\xi \equiv - \sum_{j=1}^K p_j$, and computing the $L^2$ norm (with respect to $\di \mu^V_\theta$) of the terms in the r.h.s. of the last formula give the claimed estimate \eqref{UN.est}.
\end{proof}

\section{Proof of Theorem \ref{LEM:STRUCT}}

\label{app:motzkin_path}

 In this appendix we prove Theorem \ref{LEM:STRUCT}.
 From the  structure of the matrix Lax matrix $L$  in \eqref{jacobi}, we immediately get
 
 $$[L^m]_{jj}(\ba,\bb) = S_{j-1}\left([L^m]_{11}(\ba,\bb)\right),$$
 where $S_{j}$ is the shift defined in \eqref{shift},
 thus we have to prove formula \eqref{eq:general_super_motzkin} just for the case $j=1$.
 
 To accomplish this result  we need to introduce the notion of  super  Motzkin path and  super   Motzkin polynomial, that  generalize the notion of    Motzkin path and     Motzkin polynomial  \cite{stanleyenumerative,VandeJeugt}.
  \begin{definition}
 	A super Motzkin path $p$ of size $m$ is a path in the integer plane $\mathbb{N}_0\times\mathbb{Z}$ from $(0,0)$ to $(m,0)$ where the permitted steps from $(0,0)$
	are: the step up $(1,1)$, the step down $(1,-1)$ and the horizontal step $(1,0)$. A similar definition applies to all other vertices  of  the path.
 \end{definition}
 
 The set of all super Motzkin paths of size $m$ will be denoted by $s\cM_m$.
 
 In order to introduce the  super  Motzkin polynomial associated to these paths we have to define their {\em weight}. This is done in the following way: to each up step  that occurs at height $k$, i.e. it joins the points $(l,k)$ and $(l+1,k+1)$, we associate the weight $a_k$, to a down step that joins the points $(l,k)$ and $(l+1,k-1)$ we associate the weight $a_{k-1}$, 
  to each horizontal step from $(l,k)$ to $(l+1,k)$   we associate the weight $b_k$.
  Since $k\in \Z$, the index of $a_k$ and $b_k$ are understood modulus $N$.
At this point we can define the total weight $w(p)$ of a super Motzkin path $p$ to be the product of weights of its individual steps. So it is a monomial in the commuting variables $(\bb,\ba) = (b_{-\wt{m}},\ldots,b_\wt m,a_{-\wt{m}},\ldots,a_{\wt m})$, where $\wt m = \floor{m/2}$. We remark that the total weight do not characterize uniquely the path. 
 We are now ready to give the definition of Motzkin polynomial:
  \begin{definition}
 	The super Motzkin polynomial $sP_m(\ba,\bb)$ is the sum of all weight of the elements of $s\cM_m$:
 	\begin{equation}
	\label{sMp}
 	sP_m(\ba,\bb) = \sum_{p\in s\cM_m} w(p) \, .
 	\end{equation}
 \end{definition} 
 We are now ready to relate the  Toda integrals  to the  super Motzkin polynomial $sP_m(\ba,\bb)$.
 \begin{proposition}
 	\label{prop:equivalence}
 	Given the   Lax matrix $L$ in \eqref{jacobi} then:
 	\begin{equation}\label{Eeq:tracce_motzkin}
 	\left[L^m\right]_{1,1}(\ba,\bb) = sP_m(\ba,\bb)\, .
 	\end{equation}
 	where  the super Motzkin  polynomial   $sP_m(\ba,\bb)$  is defined  in \eqref{sMp} and $a_j \equiv a_{j \mod N}, \, b_j \equiv b_{j \mod N}$.
 \end{proposition}
 \begin{proof}
 	
 	In general we have that:
 	
 	\begin{equation}
 	\label{eq:general_trace}
 		\left[L^m\right]_{1,1} = \sum_{\bj \in \mathbb{N}^{m-1}} L_{1,j_1}L_{j_1,j_2}\ldots L_{j_{m-1},1}
 	\end{equation}
	To every element of the sum we associate the path with vertices
	\[
	(0,0),\; (1,\wt{j_1}-1,), \;(2,\wt{j_2}-1),\dots, (\ell,{\wt{j_\ell}}-1),\dots, (m-1,{\wt j_{m-1}}-1), \;(m,0)
	\]
	where
	\[
	\wt{j_k}=
	\begin{cases}
	j_k&\mbox{if $j_k<\wt{m}$}\\
	j_k-N& \mbox{if $j_k\geq \wt{m}$}
	\end{cases}
	\]
	This is a super Motzkin path $p_{\bj}$ and we can associate the weight  $w(p_{\bj})$  as in  the description above therefore we have
	\[
	L_{1,j_1}L_{j_1,j_2}\ldots L_{j_{m-1},1}=w(p_{\bj})
	\]
	This is clearly a bijection. The sum of the weights of  all possible super Motzkin  paths,  is defined to be the super Motzkin  polynomial $sP_m(\ba,\bb)$
	and  thus we get the claim.
 \end{proof}

	Proceeding as in \cite[Proposition 1]{VandeJeugt}, we are  able to prove the following result, which together with Proposition \ref{prop:equivalence} proves  Theorem \ref{LEM:STRUCT}:

 \begin{proposition}
 	The super Motzkin polynomial of size $m$ is given explicitly as
	 	\begin{equation}\label{Eeq:super_motzkin}
 	sP_m (\ba,\bb)= \sum_{(\bn,\bk)\in \cA^{(m)}} \rho^{(m)}(\bn,\bk) \prod_{i = -\wt m }^{\wt m} a_i^{2n_i}  b_i^{k_i}
 	\end{equation}
 	
 where $\cA^{(m)}$ is the set  
 \begin{equation}
 \label{EcAm}
 \begin{split}
 \cA^{(m)} := \Big\{(\bn,\bk) \in \N^{m}_0 \times \N^{m}_0 \ \colon \ \ \ 
 & \sum_{i= -\wt m }^{\wt m} \left(2n_i + k_i\right) = m  , \\
 & \forall i \geq 0, \ \ \ n_i = 0 \Rightarrow n_{i+1} = k_{i+1} = 0,  \,  
 \\
 & \forall i < 0, \ \ \ n_{i+1} = 0 \Rightarrow n_{i}= k_i = 0 
 \Big\}
 \end{split}\, , 
 \end{equation}
 
 where $\wt m = \floor{m/2}$ 
 and $\rho^{(m)}(\bn, \bm) \in \N $ is 
 given by 
 \begin{align}
 \label{Erhom}
 \rho^{(m)}(\bn,\bk) := &\binom{n_{-1} + n_0 + k_0}{k_0}\binom{n_{-1} + n_0}{n_0}
 \prod_{i=-\wt m \atop i \neq -1}^{ \wt m-1 }\binom{n_i + n_{i+1} +k_{i+1} -1}{k_{i+1}}\binom{n_i + n_{i+1} -1}{n_{i+1}} \, .
 \end{align}	
 \end{proposition}

 \begin{proof}
 For a give super Motzkin path  $p$  starting at $(0,0)$ and finishing at $(0,m)$  let $k_i$ be  the number 
 of horizontal steps at height $i$ and let $n_i$ be the number of step up from height $i$ to $i+1$.   We remark the number $n_i$ of step up from height $i$ to $i+1$ is equal to 
 the number of step down from $i+1$ to $i$.
 We define the vectors  $\bk = (k_{-\wt m} , k_{-\wt m +1}, \ldots , k_{\wt m})$  and $\bn = (n_{-\wt m} , n_{-\wt m+1}, \ldots , n_{\wt m})$  and we associate
 the product
 \[
 \prod_{i = -\wt m }^{\wt m} a_i^{2n_i}  b_i^{k_i}.
 \]
 Next we need to sum over all possible  super Motzkin path  $p$   of length $m$ connecting $(0,0)$ to $(0,m)$.
  	Since the number of steps up is equal to the number of steps down, one necessarily have 
 	\[
	 \sum_{i= -\wt m }^{\wt m} \left(2n_i + k_i\right) = m\,.
 	\]
	Furthermore since the path is connected it follows that   it is not possible to have a vertex at height $i+1$ without have a vertex at height $i>0$ and the other way round if $i<0$.
	Therefore  one has
	\begin{align*}
	 & \forall i \geq 0, \ \ \ n_i = 0 \Rightarrow n_{i+1} = k_{i+1} = 0,  \,  
 \\
 & \forall i < 0, \ \ \ n_{i+1} = 0 \Rightarrow n_{i}= k_i = 0\,.
\end{align*}	
This proves the definition of the set $\cA^{(m)}$ in \eqref{EcAm}.
The final step of the proof is to count the number of paths  associated to the vectors $\bk = (k_{-\wt m} , k_{-\wt m +1}, \ldots , k_{\wt m})$  and $\bn = (n_{-\wt m} , n_{-\wt m+1}, \ldots , n_{\wt m})$.  We want to show that this number is equal to  $\rho^{(m)}(\bn,\bk) $.

	A horizontal step at height $i$ can occur just after  a  step up to  height $i$, another
 	horizontal step at height $i$, or a step  down  to  height $i$.
	This leaves a total of $n_i + n_{i+1}$  different positions at which a
 	horizontal step at height $i$ can occur. Since we have $k_i$ of horizontal  steps, the number of different
 	configurations with these step counts is the number of ways to choose $k_i$ elements from a
 	set of cardinality  $n_i + n_{i+1}$ with repetitions allowed, i.e.$\binom{n_i + n_{i+1} +k_i -1}{k_i}$.

	The number of different configurations with $n_i$ steps at height $i$ and
$n_{i+1}$ at height $i+ 1$  is given by the number of multi-sets of cardinality
$n_{i+1}$ taken from a set of cardinality $n_i$ and this number is equal to 
	$\binom{n_i + n_{i+1} -1}{n_{i+1}}. $
	
	For the horizontal steps at height 	$0$, they can also occur at the beginning of the path, this increase the number of possible positions by $1$, so the number of these configurations with these steps counts is $\binom{n_0 + n_{-1} + k_0}{k_0}$.   In this way we have obtained the coefficient $ \rho^{(m)}(\bn,\bk) $.

 \end{proof}

{\noindent \bf Acknowledgments.}
This manuscript was initiated during the research in pairs that took place in May 2019 at the
Centre International des Rescontres math\'ematiques (CIRM), Luminy, France  during the chair Morlet semester "Integrability and randomness in mathematical physics".
 The authors  thank CIRM for the   generous support and the hospitality.
This project has also  received funding from the European Union's H2020 research and innovation program under the Marie Sk\l owdoska--Curie grant No. 778010 {\em  IPaDEGAN}. \\
Finally we are grateful to Miguel Onorato  for very useful discussions.

\end{document}